\documentclass[pdflatex,sn-basic,Numbered]{sn-jnl}

\usepackage{amsmath,amssymb,amsfonts}
\usepackage{amsthm}
\usepackage{graphicx}
\usepackage{booktabs}
\usepackage{multirow}
\usepackage{tabularx}
\usepackage{array}
\usepackage{algorithm}
\usepackage{algpseudocode}
\usepackage{xcolor}
\usepackage{xurl}
\hypersetup{
  hypertexnames=false,
  pdftitle={Hybrid privacy-aware semantic search: SVD-truncated document geometry and CKKS-encrypted query reranking under a restricted threat model},
  pdfauthor={Sergey Kurilenko},
  pdfsubject={Privacy-aware dense retrieval with SVD, product quantization, and CKKS block-SIMD reranking},
  pdfkeywords={semantic search, dense retrieval, CKKS, homomorphic encryption, product quantization, leakage evaluation}
}

\newcolumntype{Y}{>{\raggedright\arraybackslash}X}

\newcommand{\R}{\mathbb{R}}
\newcommand{\doi}[1]{doi: \nolinkurl{#1}}
\newcommand{\eprint}[2][]{\nolinkurl{#2}}

\theoremstyle{plain}
\newtheorem{proposition}{Proposition}

\begin{document}

\title[Hybrid privacy-aware semantic search]{Hybrid privacy-aware semantic search: SVD-truncated document geometry and CKKS-encrypted query reranking under a restricted threat model}

\author*[1]{\fnm{Sergey} \sur{Kurilenko}}\email{sergkurilenko@gmail.com}

\affil*[1]{\orgname{Moscow Institute of Physics and Technology},
  \orgaddress{\city{Dolgoprudny}, \state{Moscow Region}, \country{Russia}}}

\abstract{Semantic search exposes a practical asymmetry: a query vector may reveal user intent, while returning exact provider vectors distributes reusable representations. We study a deliberately restricted hybrid point. A public corpus-fitted SVD basis and PQ index support client-side candidate selection; candidate IDs are revealed, whereas the projected query is encrypted with CKKS. A block-SIMD layout scores 100 plaintext candidate vectors and returns one ciphertext. At 672 dimensions, it reduces median server time from 1689.1 to 224.5 ms and response size by a factor of 99.5; a provider-only saturation test with pre-encrypted requests reaches 25.99 requests/s with 16 workers. In a frozen post-exploratory revision-analysis subset of 3,235 canonical BEIR queries, supported by 3,230 actual-CKKS records and five zero rows, five of six collections meet a $\pm0.002$ nDCG@10 equivalence rule against plaintext reranking of the identical shortlist, while ArguAna is inconclusive. SVD also outperforms seeded random projection and coordinate truncation in the tested controls. The negative results are equally important: exposed candidate sets are almost perfectly linkable in the disjoint-view audit and reproduce part of the exact neighbourhood, the public PQ artifact reveals approximate corpus geometry, and the implementation applies no output sanitization or circuit-privacy mechanism. The system therefore hides numerical query slots only conditionally on disclosed identifiers and metadata; it does not provide semantic-query, document, or access-pattern privacy.}


\maketitle

\section{Introduction}\label{sec:introduction}

Dense retrieval looks simple when described as a pipeline: an encoder maps text to vectors, an approximate index produces a shortlist, and a final inner product reranks it. Deployment is less simple because the vectors are not neutral bookkeeping. A query embedding is a numerical trace of user intent, while a provider's document embeddings form a reusable, model-derived index. Depending on the model and data, exact embeddings can support inversion and attribute inference \cite{Song2020EmbeddingLeakage,Li2023GEIA,Morris2023Vec2Text}. Protecting one side by simply handing over the other therefore misses the practical tension.

We study a deliberately narrow compromise between plaintext vector delivery and access-private retrieval. The provider publishes a lossy SVD-projected, product-quantized (PQ) candidate index. An authorized client searches that artifact locally, reveals the selected identifiers, encrypts its projected query under the Cheon--Kim--Kim--Song (CKKS) approximate homomorphic encryption scheme \cite{Cheon2017CKKS}, and decrypts the returned scores. The provider retains the full-precision projected vectors and uses them as plaintext operands. In other words, the interface avoids directly distributing the exact candidate vectors, but it does not make the corpus cryptographically private.

This asymmetry is what the title calls \emph{privacy-aware}. Numerical query values are hidden from an honest-but-curious provider, conditional on the candidate identifiers and ordinary metadata that the same request reveals. The public SVD basis and PQ codes are compression, not encryption \cite{Jegou2011PQ}; they expose approximate geometry. Repeated candidate sets may reveal semantic structure, and an adaptive client can turn exact scores into a linear measurement oracle. We therefore make no claim of document privacy, access-pattern privacy, unlinkability, or malicious-client resistance. The phrase \emph{CKKS-encrypted query reranking} refers to the encrypted query operand and returned scores; provider vectors remain plaintext on the provider host.

Once this boundary is fixed, the engineering question becomes concrete: can the provider score a useful shortlist without paying for and returning one ciphertext per candidate? Our layout repeats the encrypted query in fixed-size blocks, places four provider vectors in one plaintext at the measured $d'=672$ setting, performs a power-of-two SIMD reduction, and sparsely gathers all 100 scores into one response ciphertext. SIMD slots, rotations, and packed linear algebra are established techniques \cite{Smart2014SIMD,Halevi2014HElib,Dathathri2019CHET,Dathathri2020EVA,Zhang2021GALA,Juvekar2018Gazelle}. The contribution lies in arranging them around a public-PQ shortlist, enforcing a public-only provider process, and evaluating the resulting request under an explicit leakage contract.

The evidence follows the path of that request. We first compare per-candidate, grouped, and one-response CKKS under identical parameters, then connect the packed kernel to a separately spawned provider with a serialized public context and a read-only million-vector store. Strong float32, float16, and int8 vector-return baselines keep the performance comparison honest. The evaluation then measures public-index reconstruction, semantic information retained by candidate IDs, adaptive score extraction, concurrency, session setup, and real loopback TCP/TLS framing. At the selected point, median server time falls from 1689.12 to 224.46 ms and the response from 13,121,661 to 131,876 bytes. In the provider-only saturation loop, 16 workers consume a pre-encrypted request pool and sustain a median 25.993 requests/s with no failures. Actual CKKS covers all 3,722 graded requests; a frozen post-exploratory subset disjoint from validation contains 3,235 canonical queries, 3,230 evaluable encrypted records, and five canonical zero rows. Against plaintext reranking of the same candidates, five of six collection-specific comparisons satisfy the $\pm0.002$ equivalence rule; ArguAna remains inconclusive. This combination of positive and negative results defines the contribution more accurately than a broad privacy label would.

\noindent\textit{Relation to earlier versions.} A short preliminary design appeared in \cite{Kurilenko2026Preliminary}, and an earlier version of this manuscript was posted on arXiv under the present title \cite{Kurilenko2026Arxiv}. The journal version changes the emphasis materially: SVD and PQ are treated as disclosed compression rather than document protection, the secret-rotation argument is removed, and the evidence is rebuilt around block-SIMD execution, role-separated processes, strong disclosure baselines, validation-separated post-exploratory revision analysis, projection and PQ controls, service saturation, candidate-ID leakage, and an independently executable parameter transcript.

\section{Background and security boundary}\label{sec:background}

\subsection{Retrieval, compression, and CKKS}\label{subsec:dense-pq}

Let a document encoder produce $d_i\in\R^d$ and a query encoder produce $q\in\R^d$. With row-wise normalization, maximum inner-product search and cosine ranking coincide. Dual-encoder retrieval separates corpus encoding from online query encoding and has become a standard architecture for open-domain search \cite{Karpukhin2020DPR}. We use the multilingual E5 base model \cite{Wang2022E5,Wang2024MultilingualE5}; the systems mechanism itself is encoder-agnostic as long as scoring is a fixed inner product.

Exact search over $N$ vectors costs $O(Nd)$ arithmetic and stores $4Nd$ bytes in float32. Product quantization partitions a vector into $M$ subspaces and replaces each subvector by a centroid identifier \cite{Jegou2011PQ}. With eight bits per subquantizer, a code uses $M$ bytes per document. Asymmetric distance computation compares the uncompressed query with quantized database vectors. Faiss provides efficient exact and approximate implementations, including GPU search, but does not add confidentiality \cite{Johnson2019Faiss}.

Our candidate artifact has $d'=672$, $M=96$, and eight-bit codes. Its serialized size is 96,688,214 bytes for one million items, compared with 2.688 GB of logical float32 projected vectors and 3.072 GB of raw 768-dimensional vectors. This compression enables local candidate selection, but it deliberately publishes a lossy representation. A PQ code reconstructs an approximate projected vector by concatenating selected centroids. Consequently, the artifact cannot be described as a protected database. Section~\ref{subsec:pq-leakage-results} measures this exposure instead of assuming it away.

The projection used before PQ is also public. A passage-only sample fits a mean $\mu$ and a $768\times672$ right-singular-vector matrix $V$. Provider vectors are
\begin{equation}
  x_i=(d_i-\mu)V,
  \label{eq:doc-projection}
\end{equation}
whereas an online query uses
\begin{equation}
  z=qV.
  \label{eq:query-projection}
\end{equation}
No test-query statistic is used. Omitting query centering is deployable and does not change the ranking induced by centering all documents, because $z^\top(d_iV)-z^\top(\mu V)$ differs from the uncentred projected score by a query-dependent constant shared by every document. Projection is a public compression and packing parameter, not a privacy transform.

Compression settles how candidates are found, but it does not settle how they can be scored without exposing the query. CKKS encrypts vectors of approximate real or complex values and supports homomorphic addition and multiplication with controlled approximation error \cite{Cheon2017CKKS}. For polynomial modulus degree $N_{\mathrm{poly}}$, a ciphertext exposes $S=N_{\mathrm{poly}}/2$ complex batching slots. Ciphertext rotations implement cyclic slot permutations when the evaluator has the corresponding Galois keys. Rescaling reduces both ciphertext scale and modulus after multiplication. Security depends on the complete ring-LWE parameter set and implementation, not on application-level latency alone; parameter reporting should follow the Homomorphic Encryption Standard \cite{Albrecht2018HEStandard}.

SIMD batching is foundational HE prior art \cite{Smart2014SIMD}. HElib, SEAL, and subsequent compilers and private-inference systems optimize rotations, tensor layouts, and linear algebra \cite{Halevi2014HElib,Microsoft2023SEAL,Dathathri2019CHET,Dathathri2020EVA}. GAZELLE and GALA demonstrate highly optimized encrypted matrix--vector and dot-product organization in stronger two-party inference settings \cite{Juvekar2018Gazelle,Zhang2021GALA}. Those systems preclude a claim that this paper invents packing, rotation trees, row layouts, or segmented reduction. Our narrower contribution is an implementation-derived layout for four plaintext candidate vectors per encrypted-query ciphertext, followed by sparse packing of scores from 25 groups into one response, integrated with a public candidate index and evaluated as a retrieval request.

The distinction between ciphertext--plaintext and ciphertext--ciphertext multiplication matters. The provider owns its vectors and stores them in plaintext, so it can encode candidates as CKKS plaintexts and multiply them by an encrypted query. This uses no ciphertext multiplication and no relinearization key. A system that also hides database vectors from the evaluator would require a different construction and threat model. We do not compare those two operations as if they supplied the same functionality.

\subsection{Relation to private retrieval}\label{subsec:private-retrieval}

Private semantic search and nearest-neighbour systems cover substantially stronger security points. Tiptoe combines client metadata, linearly homomorphic nearest-neighbour computation, and PIR; it reports 360 million pages across 45 servers, 56.9 MiB communication, and 2.7 s end-to-end latency \cite{Henzinger2023Tiptoe}. Compass traverses an encrypted HNSW through tailored ORAM and protects data, query, returned identifiers and scores, and access patterns in its stated model; it evaluates the 8.84-million-passage MS MARCO corpus and reports about 1.3 s perceived latency at a selected recall point \cite{Zhu2025Compass}. These guarantees are strictly stronger than exposing our client-selected candidate identifiers.

SANNS combines lattice HE, distributed ORAM, garbled circuits, and approximate top-$k$ and reaches 10 million database entries \cite{Chen2020SANNS}. A replicated-server private ANN design obtains sublinear communication under a non-collusion assumption \cite{ServanSchreiber2022PrivateANN}. Pacmann uses client-local graph search and PIR-fetched subgraphs, with preprocessing hints and evaluations up to 100 million SIFT vectors and a 3.2-million-document MS MARCO setting \cite{Zhou2025Pacmann}. Panther co-designs PIR, secret sharing, garbled circuits, and HE for single-server private ANN and reports a 10-million-point configuration \cite{Li2025Panther}. These works rule out ``first private ANN,'' ``first single-server private ANN,'' and ``first million-scale private search'' claims.

PIR can hide which database record is downloaded, but does not by itself perform semantic scoring. SealPIR, SimplePIR/DoublePIR, and OnionPIR quantify different single-server record-retrieval trade-offs \cite{Angel2018SealPIR,Henzinger2023SimplePIR,Mughees2021OnionPIR}. Composing a PIR or ORAM candidate-fetch mechanism with our reranker is future work; the current protocol sends candidate IDs in the clear inside an authenticated transport. Secure kNN has a separate history, from transformed scalar-product databases \cite{Wong2009SecureKNN} to two-party processing of encrypted databases. Kim et al. use Paillier, two non-colluding clouds, encrypted kd-tree pruning, and parallel processing \cite{Kim2022PrivateKNN}. Their database and access-pattern goals are stronger, while our setting is one provider, high-dimensional text vectors, an exposed public artifact, and approximate CKKS arithmetic.

Table~\ref{tab:closest-work} makes the weaker leakage point explicit. ``Hidden'' summarizes the cited system's stated model; it is not a new security analysis. Latencies and scales are not normalized across hardware or workloads and therefore serve only to locate functionality, not to rank systems.

\begin{table*}[t]
\caption{Positioning against closest private-search systems. Ours intentionally exposes the public PQ artifact and candidate IDs. Published scales and costs are descriptive, not directly comparable benchmarks.}\label{tab:closest-work}
\scriptsize
\begin{tabular}{@{}>{\raggedright\arraybackslash}p{0.12\textwidth}
  >{\centering\arraybackslash}p{0.08\textwidth}
  >{\raggedright\arraybackslash}p{0.10\textwidth}
  >{\raggedright\arraybackslash}p{0.13\textwidth}
  >{\raggedright\arraybackslash}p{0.20\textwidth}
  >{\raggedright\arraybackslash}p{0.19\textwidth}@{}}
\toprule
System & Servers & Query value & Candidate and access leakage & Client-visible provider artifact & Reported setting \\
\midrule
Tiptoe \cite{Henzinger2023Tiptoe} & 45 & hidden & hidden by private computation/PIR & client centroid metadata & 360M pages; 2.7 s; 56.9 MiB \\
Compass \cite{Zhu2025Compass} & 1 & hidden & ORAM-protected & encrypted index; client state & 8.84M MS MARCO; about 1.3 s \\
Pacmann \cite{Zhou2025Pacmann} & 1 & hidden & PIR subgraph fetch & preprocessing hint/local graph state & up to 100M SIFT; 3.2M text setting \\
Panther \cite{Li2025Panther} & 1 & hidden & private protocol & output-minimizing private ANN & 10M points; 18 s; 284 MB \\
SANNS \cite{Chen2020SANNS} & $>1$ & hidden & ORAM protocol & no public vector index & 10M entries \\
Kim et al. \cite{Kim2022PrivateKNN} & 2 & hidden & designed to hide accesses & encrypted low-dimensional database & 30K synthetic and 28,056 UCI items \\
Ours & 1 & enc. given IDs & \textbf{revealed} & \textbf{public PQ artifact} & 1M-vector systems stress test; 323.5 ms local-IPC p95 \\
\bottomrule
\end{tabular}
\end{table*}

Recent supercomputing studies provide a complementary systems lens. Wang et al. study encrypted knowledge-graph semantic search for internet-of-vehicles settings \cite{Wang2025EncryptedSemanticIoV}, while Kim et al. protect an encrypted low-dimensional kNN database with a stronger two-cloud design \cite{Kim2022PrivateKNN}. Souror et al. address exact, partial, and fuzzy keyword search through searchable symmetric encryption \cite{Souror2024SmartGridSSE}; those trapdoors should not be conflated with dense inner-product retrieval. Espinoza et al. compare homomorphic matrix multiplication across SEAL, HElib, and OpenFHE and show why parameter, error, memory, and build reporting matter \cite{Espinoza2026HEMatMul}. Our study complements that work with retrieval-specific request accounting: candidate quality, serialized objects, process roles, actual encrypted scores, and a disclosed access pattern. All HE measurements still use one physical host and the SEAL/TenSEAL backend. The Windows/WSL pair is a same-host portability check, so we make no cross-hardware or cross-backend generalization.

Taken together, this literature leaves a useful space between plain vector delivery and fully access-private retrieval. That space is narrower than private search in the usual cryptographic sense, but it is also simpler to deploy. We make the compromise explicit by following the two parties, their assets, and the information visible during a request.

\subsection{Roles, functionality, and leakage}\label{subsec:roles}

The \emph{provider} owns the corpus, payloads, full-precision projected embeddings, and online reranking process. It is willing to publish a compressed PQ artifact but not to distribute the complete exact embedding table through the prescribed API. The \emph{query client} is authorized to search and receive permitted result payloads. It runs the query encoder and public-PQ search, creates the CKKS key pair, retains the secret key, and decrypts scores. The network is untrusted; TLS remains necessary for authentication, integrity, and protection of the otherwise visible identifiers and metadata. CKKS does not replace TLS.

The cryptographic adversary is an honest-but-curious provider process. It follows the circuit, lacks the secret key, and inspects its transcript. We condition the numerical query-value claim on the candidate identifiers and public metadata already revealed in that transcript. The client is not modeled as malicious-resistant: it owns the decryption key and legitimately receives selected scores. A compromised provider host sees its plaintext vector store. Malicious ciphertexts, result substitution, rollback, denial of service, client/provider collusion, endpoint compromise, and forward secrecy are outside the evaluated functionality.

The easiest way to state the intended functionality is to say exactly what crosses the boundary. For a projected query $z$ and client-selected set $C_K=(i_1,\ldots,i_K)$, the ideal numerical functionality returns
\begin{equation}
  F(z,C_K;X)=\left(z^\top x_{i_1},\ldots,z^\top x_{i_K}\right)
  \label{eq:functionality}
\end{equation}
to the client and reveals $C_K$ to the provider. The client sorts these scores and requests authorized payloads. The protocol does not hide the final payload-fetch identifiers if the same provider serves that step. Padding can hide $K$ only to a public bucket, and random permutation can hide candidate order, but neither hides membership of the set.

The provider does not directly transmit $x_i$ under this interface. That is useful when exact vectors are treated as a bulk-distribution asset. It is not cryptographic provider confidentiality: the client has a public approximate reconstruction and an exact linear measurement oracle. The distinction also explains why return-vectors must be a first-class baseline rather than silently excluded.

This functionality is useful only if its omissions are read alongside its disclosures. Table~\ref{tab:leakage} therefore acts as the contract for the rest of the paper. In particular, fresh CKKS randomness changes ciphertext bytes but does not stop repeated candidate sets from making requests linkable.

\begin{table*}[!tp]
\caption{Explicit leakage contract. ``Not provided'' means absent from the prescribed interface, not hidden against every possible attack.}\label{tab:leakage}
\small
\begin{tabularx}{\textwidth}{@{}l Y Y Y@{}}
\toprule
Observer/channel & Protected or not directly provided & Revealed by design & Explicit non-claim \\
\midrule
Provider request view & numerical slots of $z$ inside CKKS & candidate IDs, their order unless permuted, $K$/padding bucket, public parameters, ciphertext count/length, timing and session metadata & semantic query privacy, unlinkability, access-pattern privacy \\
Provider evaluation & plaintext scores remain encrypted; no decryption key in server object & circuit shape, failures, group occupancy, response size, and an unsanitized evaluated ciphertext & circuit privacy, server-operand privacy from a key-holding client, or malicious-server integrity \\
Client onboarding & raw full-precision table not directly provided & encoder and projection metadata, PQ codebooks/codes, IDs, approximate geometry and version & document privacy, database secrecy, inversion resistance \\
Client result view & no additional protected server value & selected IDs, decrypted exact scores, evaluated ciphertext, authorized payloads and timing & resistance to adaptive linear vector extraction or malicious-client inspection \\
External network & ciphertext numerical values; metadata only when TLS protects it & sizes, timing, and endpoints; without TLS, IDs and protocol metadata & protection of the full transcript by CKKS alone \\
Provider compromise & nothing on document side & full exact vectors, payloads, logs, PQ artifact & encryption at rest or document privacy \\
Client/key compromise & nothing after key loss & secret key, retained ciphertexts, scores and local index & forward secrecy supplied by CKKS \\
\bottomrule
\end{tabularx}
\end{table*}

There is a second cryptographic boundary on the return path. Ordinary CKKS input confidentiality against the evaluator does not automatically imply \emph{circuit privacy}: a key-holding recipient of an evaluated ciphertext may be entitled to the output without being entitled to learn anything else about the circuit or its plaintext operands \cite{Li2021ApproximateSecurity,Kluczniak2023CircuitPrivacy}. The implementation performs no output sanitization, noise flooding, or circuit-private bootstrapping. This omission is consistent with the restricted model because malicious clients and document-vector secrecy are already non-goals, but it must still be visible in the contract. Adding a fresh encryption of zero without a parameterized proof would not by itself justify a stronger claim.

The public artifact is another particularly important non-goal. On a deterministic 100,000-document sample, its reconstructed-and-lifted vectors have mean cosine 0.96245 with the original 768-dimensional vectors. This high similarity does not prove text recovery, but it decisively rules out treating PQ as a cryptographic shield. Candidate identifiers also carry semantic structure before any encrypted score is evaluated; Section~\ref{subsec:pq-leakage-results} measures how well the provider can reconstruct and link requests from that disclosed view. Exact scores are then linear observations: if an adaptive authorized client can force a fixed candidate to appear for $d'$ linearly independent projected queries, the client can solve for that vector. Rate limits or candidate authorization can raise the cost of extraction, but they are policies rather than guarantees of the present protocol.

\section{Design and methodology}\label{sec:design}

\subsection{Offline preparation and the online request}\label{subsec:offline}

The provider fixes an encoder version, text preprocessing, pooling rule, normalization, projection, vector dtype, and similarity definition. This avoids a common but invalid shortcut in which test-query statistics are used to center or tune the index. The measured million-vector artifact uses cached 768-dimensional E5 vectors, a passage-only projection basis, and the deployable transformations in Eqs.~\eqref{eq:doc-projection}--\eqref{eq:query-projection}. The basis was fitted to 200,000 provider vectors sampled with seed 0; randomized SVD used seed 42. Its recorded SHA-256 identifier is \texttt{48e95db206d77eea1f2f315c4889c15ac5d3fe00633851c7cdbbf501a071e87c}.

The provider then creates two stores. The exact store is a read-only float32 matrix $X\in\R^{N\times d'}$ used only by the provider process. The measured $1{,}000{,}000\times672$ NPY object occupies 2,688,000,128 bytes including its header. The public candidate artifact is a Faiss \texttt{IndexPQ} with $M=96$ eight-bit subquantizers; it occupies 96,688,214 serialized bytes and has SHA-256 \texttt{479df1b07e1cd3d374923f2cf566af4e2f84af98428f4fc0b1494679d23a45dc}. Every code consumes 96 bytes, and the artifact size grows linearly with corpus size. Codebook training and projection fitting are offline operations; an index-version change requires the client and provider to agree on the new artifact before querying.

For each client key epoch, the client creates a private CKKS context and transmits only serialized encryption/evaluation material: encryption parameters, a public key, and a limited Galois-key set. The provider deserializes those values into a public context. No relinearization key is required because the online circuit uses ciphertext--plaintext multiplication only. Context setup is not counted in warm request latency. The client may cache the public context at the provider for the epoch, but multi-client cache cost and key rotation are not evaluated here.

An online request begins entirely on the client. It encodes and normalizes the user's text, computes $z=qV$, and searches the public PQ artifact for $C_K$. The measured implementation preserves PQ order for deterministic auditing and sends the resulting $K$ 64-bit identifiers. The client then pads $z$ to length $L$, repeats it $B$ times across the CKKS slots, and encrypts the repeated vector once with fresh randomness.

The provider uses those identifiers to gather the corresponding exact rows, arranges them in groups of $B$, and evaluates the packed score circuit with only the public context. After the block reductions, it masks the score slots, shifts group outputs to non-colliding positions, and returns one serialized ciphertext. The client decrypts the designated slots, restores their association with candidate IDs, ranks locally, and requests the authorized payloads. Candidate IDs and the encrypted query travel together, so a deployment still needs authenticated transport. Our byte counts cover the application objects themselves and exclude TLS, HTTP, operating-system pipe, and multiprocessing headers.

\subsection{Block-SIMD kernel, process separation, and correctness}\label{subsec:kernel}

With the leakage boundary fixed, the remaining bottleneck is the apparent need to compute and serialize $K$ encrypted inner products. The block layout below is designed to share the expensive reduction across several candidates and to leave the client with a single response object.

Let $S=N_{\mathrm{poly}}/2$ be the number of CKKS slots and
\begin{equation}
  L=2^{\lceil\log_2 d'\rceil},\qquad B=\left\lfloor\frac{S}{L}\right\rfloor,
  \qquad m=\left\lceil\frac{K}{B}\right\rceil.
  \label{eq:layout-counts}
\end{equation}
For a group $g$ of at most $B$ candidates, define zero-padded vectors
$\bar z,\bar x_i\in\R^L$ and slot vectors
\begin{align}
  Q &= \bar z\,\|\,\bar z\,\|\cdots\|\,\bar z, \\
  P_g &= \bar x_{gB}\,\|\,\bar x_{gB+1}\,\|\cdots\|\,\bar x_{gB+B-1}.
  \label{eq:slot-vectors}
\end{align}
Missing entries in the last group are zero. The client encrypts $Q$ once; each $P_g$ is a provider plaintext.

After a ciphertext--plaintext product and one rescale, rotations by
$1,2,4,\ldots,L/2$ and in-place additions form a power-of-two reduction. Because the client reads only the first slot of each $L$-slot block, cross-boundary values elsewhere need not be retained. There are $\log_2L$ rotations and additions per group. This first stage produces $B$ useful values per group ciphertext, at indices $0,L,2L,\ldots,(B-1)L$.

Returning these $m$ ciphertexts already reduces the naive $K$ outputs by a factor approaching $B$. Our final stage reduces the response to one ciphertext. For group $g$, the provider plaintext-multiplies by a sparse mask that is one at the $B$ score positions and zero elsewhere. It cyclically shifts that ciphertext by $g$ slots, synthesized from the binary decomposition of $g$, and adds it to an accumulator. The client records the resulting indices $(bL-g)\bmod S$. The constraints $K\leq S$ and $m\leq L$ ensure that the implemented mapping has room for one score per candidate without group-offset collision. Algorithm~\ref{alg:kernel} gives the exact structure.

\begin{algorithm}[t]
\caption{Provider-side one-response block-packed reranking}\label{alg:kernel}
\begin{algorithmic}[1]
\Require public CKKS context, encrypted repeated query $c_Q$, candidates $X[C_K]$, $L$, $B$
\State $c_{\mathrm{out}}\gets\bot$
\For{$g=0,\ldots,\lceil K/B\rceil-1$}
  \State encode up to $B$ candidates as $P_g$ in $L$-slot blocks
  \State $c_g\gets\mathrm{Rescale}(c_Q\odot P_g)$
  \For{$r\in\{1,2,4,\ldots,L/2\}$}
    \State $c_g\gets c_g+\mathrm{Rot}(c_g,r)$
  \EndFor
  \State $c_g\gets c_g\odot M_g$ \Comment{sparse score-slot mask; no second rescale}
  \For{set bit $r$ in the binary representation of $g$}
    \State $c_g\gets\mathrm{Rot}(c_g,r)$
  \EndFor
  \State $c_{\mathrm{out}}\gets c_g$ if $g=0$, else $c_{\mathrm{out}}+c_g$
\EndFor
\State \Return serialized $c_{\mathrm{out}}$ and the public slot map
\end{algorithmic}
\end{algorithm}

The measured setting is $N_{\mathrm{poly}}=8192$, hence $S=4096$; $d'=672$ gives $L=1024$ and $B=4$; $K=100$ gives $m=25$. The coefficient-modulus bit sizes are $[60,40,60]$, the initial scale is $2^{40}$, and the ten Galois rotation steps are $1$ through 512 in powers of two. Each primary product has nominal scale $2^{80}$ and is rescaled across the 40-bit middle prime to approximately $2^{40}$. The final sparse masks are encoded at $2^{19}$; their products are approximately $2^{59}$ and are deliberately not rescaled, preserving the last level for rotations and additions.

Table~\ref{tab:ckks-parameters} records the exact key-context values rather than only the requested bit sizes. The three primes are extracted by calling the same \texttt{CoeffModulus.Create} API used by the prototype. Their product has $\log_2 q=159.9999998065$. At degree 8192, SEAL's \texttt{TC128} table permits 218 coefficient-modulus bits, leaving 58 bits of table headroom; this library guard is a parameter check, not an independent attack estimate. Source inspection at the pinned Microsoft SEAL 4.1.1 revision finds independent uniform ternary secret coefficients and exact centred-binomial error with 21 Bernoulli bits per side, support $[-21,21]$, and standard deviation $\sqrt{10.5}=3.24037$.

\begin{table*}[t]
\caption{Exact CKKS key-context and distribution record. Evaluation-key assumptions and independent attack estimates are discussed separately because neither follows from the \texttt{TC128} table check.}\label{tab:ckks-parameters}
\centering
\small
\begin{tabularx}{\textwidth}{@{}lY@{}}
\toprule
Item & Recorded value \\
\midrule
Polynomial degree / complex slots & 8192 / 4096 \\
Coefficient-modulus primes & 1152921504606748673; 1099511480321; 1152921504606830593 \\
Total modulus / data levels & 160 bits; first data level 100 bits; last data level 60 bits \\
Scale and masks & initial $2^{40}$; sparse final masks encoded at $2^{19}$ \\
Secret distribution & independent uniform ternary coefficients in $\{-1,0,1\}$ \\
Error distribution & centred binomial $\mathrm{CBD}(21)$; $\sigma=3.24037$; Gaussian-$3.2$ sensitivity model \\
Evaluation keys & public key and ten selected Galois keys; no relinearization keys \\
Source pins & TenSEAL 0.3.16 binary SHA-256 \nolinkurl{7d17b5864a9929c23c78e9626271f55b7a52e4deda51cfe096f474e7a40a1716}; SEAL source commit \nolinkurl{206648d0e4634e5c61dcf9370676630268290b59} \\
\bottomrule
\end{tabularx}
\end{table*}

The Galois keys are encryptions or key-switching functions of automorphisms of the secret key. Their use therefore retains the standard circular/KDM assumption; the LWE cost estimator does not prove that assumption. It also does not establish circuit privacy for the returned ciphertext. Following the estimator's requested citation \cite{Albrecht2015LWEHardness}, the artifact pins malb's implementation at commit \nolinkurl{3e48ef421ec256afddb3e7d2249a77eab6e9ba12}, uses an unbounded-sample sensitivity model, and archives every Sage input and raw transcript.

The estimator results are reported by cost model rather than collapsed into one security number. Under the ADPS16/GSA rough model, uSVP, dual-hybrid, and plain dual complete at $\log_2(\mathrm{rop})=147.8$, 146.9, and 148.3 for exact CBD(21); the Gaussian-3.2 sensitivity model gives the same rounded exponents. The upstream buffered rough batch reaches its 3600-second limit inside bounded-noise Arora--GB, and the isolated Arora--GB point reaches 1800 seconds without returning a cost. Thus 146.9 is the minimum among the completed rough attacks, not an exhaustive all-attack minimum. Under the estimator's default MATZOV/GSA model, the minimum completed cost is 175.8 for both error models (BDD and BDD-hybrid); uSVP and dual-hybrid give 176.3, plain dual 177.8, and BDD-MITM-hybrid 247.0--247.6. Both coded-BKW points reach their separate 900-second limits without a cost. These model-dependent heuristic estimates exceed 128 bits for every completed attack, but the timeout records, circular/KDM assumption, and implementation side channels remain outside that statement.

Table~\ref{tab:operations} reports counts from the implemented loops, not an asymptotic idealization. The one-response method uses more plaintext multiplications than the 25-output grouped variant because of masking, but serialization falls from 25 objects to one. For $g=0,\ldots,24$, the sum of binary Hamming weights is 54; hence score packing adds 54 rotations and 24 accumulator additions. No ciphertext--ciphertext multiplication or relinearization occurs.

\begin{table}[t]
\caption{Implemented operation counts for $d'=672$, $K=100$, $L=1024$, $B=4$, and $m=25$. Counts exclude encoding and serialization.}\label{tab:operations}
\centering
\small
\begin{tabular}{@{}lrrrrr@{}}
\toprule
Method & ct--pt mult. & Rescale & Rotations & Additions & Output ct \\
\midrule
Per-candidate & 100 & 100 & 1000 & 1000 & 100 \\
Grouped blocks & 25 & 25 & 250 & 250 & 25 \\
One-response & 50 & 25 & 304 & 274 & 1 \\
\bottomrule
\end{tabular}
\end{table}

Packing alone would not establish the claimed role boundary if the same process could still call \texttt{decrypt}. The implementation therefore defines distinct client and server context classes. The client class owns the secret key and decryptor. The server class owns only the SEAL context, encoder, evaluator, public key, and ten selected Galois keys. Serialization includes a \texttt{contains\_secret\_key=false} header and exactly three public payloads: parameters, public key, and Galois keys. Deserialization rejects any envelope that does not declare the public-only role.

The unified benchmark uses the Windows \texttt{spawn} start method rather than forking client memory. The child process receives only a pipe endpoint at process creation. Initialization then sends the serialized public context, the exact-matrix path, the public kernel name, and its packing parameter over the pipe. The server opens the projected NPY matrix read-only. Per-query messages contain a serialized encrypted query and uint64 candidate IDs; replies contain a serialized encrypted score response and phase timings.

The recorded server attestation reports a distinct process identifier, process name \texttt{public-ckks-rerank-server}, class \texttt{ServerCKKSContext}, \texttt{context\_is\_public=true}, \texttt{has\_secret\_key=false}, and no secret-key or decryptor attribute. It received 7,694,330 public-context bytes and opened a read-only $1{,}000{,}000\times672$ float32 matrix. This is a concrete structural and negative-capability test; it is not remote attestation, a sandbox proof, or protection against provider-host compromise.

This organization also makes the cost boundary easy to see. Client-side exhaustive PQ scanning is linear in $N$ and $M$ for the measured \texttt{IndexPQ}. The encrypted reranker is independent of $N$ after candidates are selected: it performs $m$ primary ciphertext--plaintext products, $m\log_2L$ reduction rotations, and a score-packing pass. Exact provider storage remains $O(Nd')$; public client storage is $O(NM)$ bytes plus codebooks. The current architecture therefore exchanges online bandwidth for a substantial persistent client artifact.

The request contains one encrypted repeated query plus $8K$ candidate-ID bytes. The final response contains one ciphertext as long as $K\leq S$ and the non-collision constraints hold. In the measured unified test, request, response, and total application payloads average 236,281, 131,876, and 368,157 bytes. The per-candidate response is 13,121,661 bytes. By contrast, returning projected float32 vectors costs $4Kd'=268{,}800$ response bytes, float16 costs 134,400, and row-wise symmetric int8 plus one float32 scale per row costs 67,600. Encrypted scores are therefore not a bandwidth-dominant solution against all vector-return formats; their purpose is to avoid directly returning candidate vectors under the prescribed interface.

The implementation still has to preserve the intended score semantics after padding, rotation, masking, and shifting. The following proposition isolates that functional claim from both the novelty and security discussion.

\begin{proposition}[Block score semantics]\label{prop:block-score}
Let $z,x_i\in\R^{d'}$, let $L\geq d'$ be a power of two, and let the vectors be zero-padded and arranged as in Eq.~\eqref{eq:slot-vectors}. In exact slot arithmetic, the designated start slot of block $b$ after the power-of-two reduction equals $z^\top x_{gB+b}$.
\end{proposition}

\begin{proof}
The ciphertext--plaintext product at positions $bL+j$, $0\leq j<d'$, is $z_jx_{gB+b,j}$, and positions $bL+j$ for $d'\leq j<L$ are zero. At the block-start slot, additions after rotations by $1,2,\ldots,L/2$ form the binary reduction tree over exactly these $L$ entries. Thus the value is $\sum_{j=0}^{d'-1}z_jx_{gB+b,j}$. Values in non-designated positions are irrelevant because the sparse mask retains only the starts. Masking, cyclic shifting, and adding group ciphertexts are linear and move, rather than alter, the retained values. CKKS evaluates an approximation to this exact arithmetic; its error is measured separately.
\end{proof}

The proposition is a functional statement, not a novelty or security claim. Unit tests compare each method to NumPy dot products and compare methods to each other. The $d'=672$, $K=100$ microbenchmark has maximum absolute errors $1.12\times10^{-6}$ for per-candidate CKKS, $8.82\times10^{-7}$ for 25-output blocks, and $1.50\times10^{-5}$ for the one-response method. In 400 real projected-query requests, the per-query maximum absolute error has mean $2.30\times10^{-5}$ (95\% bootstrap interval $[2.27,2.32]\times10^{-5}$), p99 $3.03\times10^{-5}$, and global maximum $3.32\times10^{-5}$.

\begin{proposition}[Top-$r$ stability in a fixed shortlist]\label{prop:ranking}
For fixed candidate set $C$, let exact scores be $s_i$ and decrypted scores be $\hat s_i$. If $\max_{i\in C}|\hat s_i-s_i|\leq\epsilon$ and the exact gap between positions $r$ and $r+1$ is greater than $2\epsilon$, then the exact and decrypted top-$r$ sets are identical.
\end{proposition}

\begin{proof}
For any exact top-$r$ item $i$ and excluded item $j$, $s_i-s_j>2\epsilon$. Therefore $\hat s_i-\hat s_j\geq(s_i-\epsilon)-(s_j+\epsilon)>0$. No excluded item can cross the boundary.
\end{proof}

The condition concerns ranking \emph{within the fixed PQ shortlist}; it says nothing about relevant documents omitted by candidate selection. In the held-out test, the complete top-10 order matches the plaintext projected-shortlist order for 96.25\% of queries (95\% bootstrap interval 94.25--98.0\%). Hit@10 is identical, while Hit@1 differs on one query. We report both IR metrics and numerical error because neither alone proves the other.

Correct scores do not by themselves imply a broad privacy guarantee.
The provider transcript contains the public context, encrypted repeated query, candidate IDs, circuit choice, sizes, and timing. Subject to the security of the instantiated CKKS encryption and excluding side channels beyond this transcript, fresh randomized encryption protects the plaintext numerical slots from the honest-but-curious provider. The claim is conditional because $C_K$ is computed from the same query and is revealed. Two numerically different queries that induce distinctive or repeated candidate sets may be semantically distinguishable through that leakage. We therefore avoid the unqualified phrase ``query privacy.''

To make that qualification empirical, the provider-side audit uses only the ordered candidate IDs and the corresponding exact vectors that the provider already owns. It never reads PQ distances, plaintext scores, a secret key, or a decryptor. Three deliberately simple estimators turn the disclosed list into a query direction: an unweighted set centroid, a logarithmically rank-weighted centroid, and a ridge least-squares fit to fixed rank targets. We evaluate $K\in\{20,50,100,200\}$ by cosine to the true projected query, overlap with its exact top-10 and top-100, recovery of judged or constructed relevant items, and linkability. The linkability test divides each candidate list into two disjoint alternating-ID views and measures how well reconstructed-view cosine separates matching from nonmatching queries. Keeping the set-only and order-aware estimators separate shows what random candidate permutation can and cannot mitigate.

The provider's candidate vectors are plaintext operands and are not protected from the provider, its process, or its storage compromise. The client receives exact approximate-CKKS scores. The protocol supplies neither simulation-based two-sided privacy nor cryptographic database confidentiality. Stronger systems such as Compass, Tiptoe, SANNS, Pacmann, and Panther address different leakage and trust points and cannot be reduced to performance baselines for this functionality.

The restriction is equally important on the client side.\label{subsec:score-oracle}

Suppose an authorized client obtains the score of one fixed projected provider vector $x\in\R^{d'}$ under projected queries $z_1,\ldots,z_t$. Writing the rows as $Z\in\R^{t\times d'}$ and scores as $y\in\R^t$ gives
\begin{equation}
  y=Zx+e,
  \label{eq:oracle-system}
\end{equation}
where $e$ contains CKKS approximation error. If $t\geq d'$ and $Z$ has full column rank, least squares recovers $x$ up to conditioning and error. At $d'=672$, 672 well-conditioned independent observations are algebraically sufficient. Candidate selection may make it difficult to hold the same item in every shortlist, but the protocol does not enforce such a restriction and thus cannot claim resistance.

The public PQ reconstruction reduces the attacker's prior uncertainty before any score query. Defenses could include rate limits, per-document query budgets, candidate authorization, returning only a limited top set, constraining queries to attested encoder outputs, adding calibrated noise, or integrating access-private retrieval. Each changes functionality, utility, or trust. None is implemented here, so none appears in the guarantee.

\subsection{Experimental protocol}\label{subsec:implementation}

The prototype is Python 3.11.15. It uses TenSEAL 0.3.16 and the distributed \texttt{tenseal.sealapi} interface to Microsoft SEAL for evaluator rotations, rather than private extension bindings. NumPy is 2.4.3 and Faiss is 1.13.2. PyTorch 2.10.0 with CUDA 12.8 supports embedding/projection reference work. The host exposes 16 logical CPU threads and an NVIDIA GeForce RTX 5060 with 8,546,484,224 bytes visible to PyTorch. CKKS evaluation is CPU-only; the GPU is not credited with HE acceleration.

The exact projected matrix is memory-mapped. Client and server are separate spawned processes connected by a local multiprocessing pipe. The original unified path uses this topology to test key separation and real serialization; its online latency excludes context generation, public-context initialization, cached query-embedding inference, TLS, and WAN/LAN transfer. The microbenchmark reports context/key generation separately at 265.76 ms and query encryption/serialization at 15.69 ms; its serialized public context and query are 7,694,632 and 235,541 bytes.

A separate systems run removes two of those earlier evidence gaps without pretending to be a WAN study. It fixes the OpenMP, MKL, OpenBLAS, NumExpr, and BLIS thread pools to one, gives every closed-loop client its own independently spawned public-only worker, and measures concurrency $1,2,4,8,$ and $16$. The saturation loop reuses a pool of already encrypted requests so that throughput measures provider evaluation and serialized pipe responses rather than client encryption or decryption. Three packed-kernel restarts alternate ascending and descending condition order; each condition executes three warmups and 20 requests per worker. Two smaller naive-kernel restarts cover concurrency one and eight. The same harness samples process CPU time and worker working sets, repeats cold and persistent full-client sessions five times, and sends the exact serialized protocol over persistent TCP and TLS~1.3 loopback sockets for 30 requests after three warmups. Loopback isolates socket, framing, and TLS record overhead; it does not emulate physical-network RTT, loss, congestion, or certificate deployment.

The coefficient-modulus chain, scales, and selected rotations are fixed across the three kernel methods. Each microbenchmark point performs two warmups; the headline $d'=672$, $K=100$ point uses 20 measured repetitions. The dimension and shortlist sweep uses seven repetitions per additional point. A separate unified run evaluates every one of 400 held-out query requests with actual CKKS. Five additional spawned-server restarts each execute the same 100 validation queries to expose process-to-process variation.

The evaluation uses two complementary tracks. The first is a scale anchor comprising one million cached Wikipedia passage embeddings of dimension 768 and 500 cached self-query embeddings. The qrel recipe associates each query with one deterministic sampled source passage using seed 42. The deterministic query split uses seed 2026: 100 queries are validation data and 400 are held out for test. Because each query has one constructed source target rather than graded human relevance judgments, Hit@1, Hit@10, MRR@10, and nDCG@10 measure self-retrieval stability only. They must not be interpreted as production semantic relevance.

The stress test nevertheless exercises the claimed systems path at full scale. Candidate search scans the one-million-entry public PQ object. The provider gathers from the full projected memmap. Actual CKKS processes all 100 candidates for every held-out query. The exact projected exhaustive search is a GPU-assisted local quality reference; its batch-amortized time is not mixed into the online client/server latency. The projected-shortlist plaintext computation is an untimed numerical reference for CKKS. The post-protocol auditor may read exact candidate rows after the request solely to calculate error; this is not a deployable client action.

The projection controls ask whether the SVD in the title contributes more than dimension truncation alone. At $d'\in\{384,672\}$, the frozen corpus-only SVD prefix is compared with the first $d'$ coordinates and with a seeded Gaussian orthoprojector. The random matrix uses seed 2026, is orthonormalized once at 672 columns, and is evaluated through nested prefixes; documents are centred by the same passage mean and online queries are not centred. Neither control sees a query label. Separate validation and revision-analysis branches are retained, and the already selected 672-dimensional operating point is not changed from these rows.

The broader dimension sweep uses nested SVD prefixes at $d'\in\{192,256,384,512,672\}$ and the full centred 768-dimensional geometry. Its pure exact-retrieval branch changes only retained dimension. A second, explicitly joint storage curve trains eight-bit PQ with $M=d'/8$, seed 42, and $K=100$, thereby fixing an eight-coordinate subvector width while code bytes grow with dimension. This joint branch must not be read as a pure SVD ablation. Around the frozen $d'=672$, $M=96$, $K=100$ deployment, sensitivity runs use $M\in\{32,48,84,96\}$ and $K\in\{20,50,100,200\}$. FiQA and TREC-COVID additionally repeat $M\in\{84,96\}$ with seeds 17, 42, and 2026; all seeds are reported without best-seed substitution.

Shortlist size is selected only on validation. We compare $K\in\{20,40,64,80,100,200\}$ and freeze $K=100$ before opening the test metrics. The validation candidate-source recall increases from 0.92 at $K=20$ to 0.96 at $K=100$ and does not improve at $K=200$. The choice trades a plateaued diagnostic recall against server work; it is not tuned on the 400 test queries.

The second track asks whether the same approximate geometry preserves graded retrieval quality on six heterogeneous BEIR collections with official qrels: SciFact, NFCorpus, ArguAna, SciDocs, FiQA-2018, and TREC-COVID \cite{Thakur2021BEIR}. The harness loads pinned corpus, query, and qrel revisions, records checksums, and never treats a smoke subset as benchmark-comparable. Inputs to \texttt{intfloat/multilingual-e5-base} use literal \texttt{query:} and \texttt{passage:} prefixes, masked-mean pooling over non-padding tokens, row-wise $\ell_2$ normalization, and maximum length 512. Passage title and body are concatenated under one fixed recipe. Encoding uses the RTX 5060; that stage is reported separately from CPU HE.

Validation follows the benchmark rather than silently borrowing from the final rows. FiQA-2018 and NFCorpus use their official development queries, while SciFact uses official training queries. ArguAna, SciDocs, and TREC-COVID have no suitable separate development split in the pinned recipe, so their official-test query IDs are partitioned deterministically into 20\% validation and 80\% revision-analysis sets with seed 2026. The frozen report uses 300 official-test SciFact, 323 official-test NFCorpus, 648 official-test FiQA-2018, 800 held-out SciDocs, 40 held-out TREC-COVID, and 1,124 canonical held-out ArguAna queries. Projection fitting and PQ training use corpus embeddings only; no query mean or qrel label enters either artifact.

The six frozen revision-analysis sets contain 3,235 canonical queries. Five ArguAna queries have judged-positive document identifiers absent from the pinned official corpus archive. They remain in the primary 1,124-query denominator with zero contribution, while the separately labeled evaluable sensitivity analysis contains 1,119 queries. The CKKS subset analysis selects 3,230 independently executed per-query records from the existing 3,722-request replay; it does not rerun them, and the implementation has no cross-query training, aggregation, or cryptographic state. The five absent-document rows contribute the same zero to both sides of the canonical utility comparison. Full-test rows that overlap validation remain in the artifact for arithmetic and systems provenance but are not used as revision-analysis utility evidence.

The harness compares raw exact retrieval, projected exact retrieval, PQ ordering, and exact projected reranking of the same PQ candidates. It computes nDCG@10 with graded gains, MRR@10, Recall@10, Recall@100, and paired 10,000-sample query-bootstrap intervals. The revision analysis plan was frozen before calculating the strict subsets, although it was necessarily written after the earlier exploratory full-test analysis and is not presented as prospective preregistration. Its operational non-inferiority margin is an absolute nDCG@10 degradation of 0.002. This is a frozen engineering tolerance, not a standard BEIR or perceptibility threshold; the raw intervals are retained so that another margin can be applied. Non-inferiority is reported only when the lower endpoint of the paired 95\% interval exceeds $-0.002$; equivalence requires the complete interval to lie inside $[-0.002,0.002]$. Merely containing zero is never treated as evidence of equivalence. Decisions are collection-specific and all six are reported; no pooled hypothesis or ``at least one success'' claim is made, so the five-of-six count is a descriptive summary rather than a multiplicity-adjusted global test.

Both tracks use the same conceptual ladder of baselines. \emph{Exact projected} retrieval exhaustively ranks the provider's projected vectors and serves only as a local utility reference, while \emph{PQ-only} exposes the ordering already available from the public index. The strongest performance alternative is \emph{return vectors}: the client sends the same candidate IDs and receives raw or projected vectors in float32/float16, or projected vectors under row-wise symmetric int8 quantization. These methods are disclosure baselines, not privacy protocols.

To isolate arithmetic effects, a plaintext projected-shortlist computation supplies paired exact scores. Naive CKKS then returns one ciphertext per candidate, whereas the grouped and one-response variants use the block-SIMD layout of Section~\ref{subsec:kernel}. The one-response variant is the deployed path.

Every online baseline uses the same frozen candidate IDs when compared on the test split. Utility reports per-query Hit@1/10, MRR@10, and nDCG@10 with 95\% paired query-bootstrap intervals from 10,000 resamples. Numerical evaluation reports maximum and mean absolute error, RMSE, and exact top-10-order match. Systems evaluation reports request-level means and p50/p95/p99, stage timings, ciphertext counts, and application bytes. Latency quantiles are computed from complete requests; we never sum stage p95 values. Five process restarts assess validation latency variability, but repeated validation queries are used only for systems latency, not as independent utility observations.

Reproducibility is part of this protocol rather than an afterthought. The artifact records raw per-query JSONL logs and machine-readable summaries for retrieval, CKKS microbenchmarks, the spawned-process path, five restarts, vector-return baselines, scaling, reconstruction, and every official-test replay. An environment manifest pins package versions, GPU visibility, hardware, and the source revision; projection and PQ objects carry shapes, seeds, paths, and hashes. Unit tests cover slot semantics, public/private context roles, malformed envelopes, deterministic splits, numerical agreement, and metric calculations.

The review bundle also freezes the BEIR result map, dataset and qrel revisions, the E5 model revision \texttt{d128750597153bb5987e10b1c3493a34e5a4502a}, cache manifests, and a command for each replay. Derived NPY, NPZ, and Faiss objects are keyed by the model recipe and corpus/query hashes. Until an immutable public tag and archival DOI are created, these supplied hashes are the reproducibility anchors.

\section{Results}\label{sec:evaluation}

\subsection{Choosing the operating point and evaluating the packed kernel}\label{subsec:validation-selection}

The two evaluation tracks should be read together. The million-document experiment stresses the complete systems path at a fixed operating point, whereas the BEIR collections test graded retrieval quality across heterogeneous tasks. We first check that the projection choice is substantive, then revisit the validation-only shortlist decision and isolate the kernel that dominates the encrypted stage.

Before turning to that shortlist, Table~\ref{tab:projection-controls} checks the projection itself on the frozen post-exploratory BEIR subsets. SVD has the highest nDCG@10 in every collection at both dimensions. At 384 dimensions its advantage over the stronger of the two non-SVD controls ranges from .0535 to .1823; at 672 dimensions the range narrows to .0084--.0255. The table does not show that SVD is universally optimal, but it rules out the easier explanation that any seeded linear truncation would have produced the reported utility.

\begin{table*}[t]
\caption{Projection controls on the frozen BEIR revision-analysis subsets. Values are exhaustive projected nDCG@10; no PQ or CKKS arithmetic is used in this table. Rnd is the seed-2026 Gaussian orthoprojector and Coord is first-coordinate truncation.}\label{tab:projection-controls}
\centering
\small
\begin{tabular}{@{}lrrrrrr@{}}
\toprule
& \multicolumn{3}{c}{$d'=384$} & \multicolumn{3}{c}{$d'=672$} \\
\cmidrule(lr){2-4}\cmidrule(lr){5-7}
Collection & SVD & Rnd & Coord & SVD & Rnd & Coord \\
\midrule
ArguAna & .43144 & .25561 & .30513 & .44178 & .42203 & .41406 \\
FiQA-2018 & .36865 & .19259 & .21876 & .38012 & .34573 & .36269 \\
NFCorpus & .32297 & .22110 & .24183 & .32703 & .31161 & .30485 \\
SciDocs & .16278 & .10013 & .10925 & .16614 & .15772 & .15763 \\
SciFact & .68846 & .55295 & .53245 & .69678 & .67697 & .66195 \\
TREC-COVID & .67220 & .48986 & .43017 & .68670 & .65199 & .66123 \\
\bottomrule
\end{tabular}
\end{table*}

The complete sweep in Table~\ref{tab:svd-pareto} turns that control into a systems choice. Retained corpus variance rises quickly and exact macro nDCG@10 begins to plateau between 512 and 672 dimensions. On validation, moving from 672 to 512 loses .00122 macro nDCG in the pure exact branch and .00078 in the joint PQ branch. In return, the PQ code falls from 84 to 64 bytes per document and CKKS padding falls from 1024 to 512 slots per candidate, doubling block occupancy from four to eight. We therefore report 512 dimensions as a latency/storage-favouring alternative and retain the already frozen 672-dimensional deployment as the utility-favouring point. The post-exploratory revision rows document robustness; they are not used to reselect the deployment.

\begin{table*}[t]
\caption{SVD dimension and joint storage Pareto macro-averaged over six BEIR collections. Exact is the pure nested-SVD branch. PQ+rerank fixes $M=d'/8$, seed 42, and $K=100$; it is therefore a joint projection/storage curve. ``Revision'' denotes the frozen post-exploratory subset.}\label{tab:svd-pareto}
\centering
\scriptsize
\resizebox{\textwidth}{!}{%
\begin{tabular}{@{}rrrrrrrrr@{}}
\toprule
$d'$ & Retained variance & Exact validation nDCG & Exact revision nDCG & PQ+rerank revision nDCG & Candidate R@100 & PQ bytes/doc & CKKS $L$ & Candidates/block \\
\midrule
192 & .7828 & .4049 & .4081 & .4092 & .4917 & 24 & 256 & 16 \\
256 & .8578 & .4306 & .4270 & .4281 & .5034 & 32 & 256 & 16 \\
384 & .9499 & .4517 & .4411 & .4435 & .5127 & 48 & 512 & 8 \\
512 & .9881 & .4560 & .4470 & .4472 & .5104 & 64 & 512 & 8 \\
672 & .9986 & .4572 & .4498 & .4490 & .5113 & 84 & 1024 & 4 \\
768 & 1.0000 & .4583 & .4512 & .4478 & .5140 & 96 & 1024 & 4 \\
\bottomrule
\end{tabular}
}
\end{table*}

Figure~\ref{fig:svd-pareto} keeps the collection-level curves visible rather than letting a macro-average hide heterogeneous behaviour. The random and coordinate controls appear only at their measured dimensions; the dashed joint curve uses the fixed eight-coordinate PQ subvector rule.

\begin{figure*}[t]
\centering
\includegraphics[width=\textwidth]{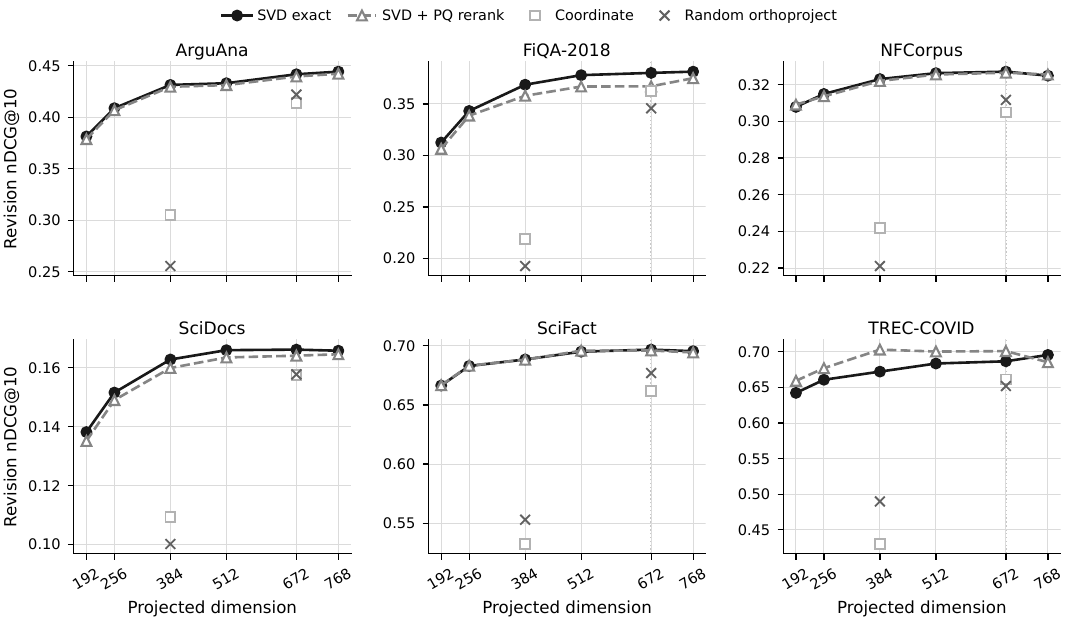}
\caption{Collection-level SVD/PQ Pareto on the frozen post-exploratory BEIR revision subsets. Solid circles are exhaustive SVD retrieval; dashed triangles add PQ candidate selection and exact projected reranking with $M=d'/8$. Coordinate and random controls are shown at 384 and 672 dimensions. The vertical guide marks the frozen 672-dimensional point; it is not selected from these revision rows.}
\label{fig:svd-pareto}
\end{figure*}

The local sensitivity analysis in Table~\ref{tab:pq-sensitivity} supports two smaller conclusions. First, $K=100$ remains a strong knee: doubling encrypted work to $K=200$ adds only .0010 macro nDCG@10, although it increases candidate coverage. Second, $M=84$ approaches the frozen $M=96$ quality while reducing code bytes per document by 12.5\%; it is a plausible storage alternative, not a retrospectively substituted winner. Across three seeds, FiQA's nDCG sample standard deviation at $K=100$ is .00136 for $M=84$ and .00208 for $M=96$. TREC-COVID's 40-query revision subset yields the larger .01094 and .01196 dispersions, while candidate-recall standard deviations remain below .0016.

\begin{table*}[t]
\caption{PQ sensitivity at $d'=672$, macro-averaged over six frozen BEIR revision subsets. Left: $K$ sensitivity at $M=96$; candidate recall uses all available $K$ results, while reranked Recall@100 uses at most 100. Right: $M$ sensitivity at $K=100$; storage sums the six serialized indices.}\label{tab:pq-sensitivity}
\centering
\small
\resizebox{\textwidth}{!}{%
\begin{tabular}{@{}rrrr rrrr@{}}
\toprule
$K$ & nDCG@10 & Candidate R@$K$ & Reranked R@100 & $M$ (bytes/doc) & nDCG@10 & Candidate R@100 & Six-index size (MB) \\
\midrule
20  & .4352 & .3962 & .3962 & 32 & .4166 & .4397 & 12.84 \\
50  & .4453 & .4643 & .4643 & 48 & .4373 & .4722 & 17.19 \\
100 & .4506 & .5148 & .5148 & 84 & .4490 & .5113 & 26.99 \\
200 & .4516 & .5735 & .5466 & 96 & .4506 & .5148 & 30.25 \\
\bottomrule
\end{tabular}
}
\end{table*}

Table~\ref{tab:k-validation} then shows the complete shortlist-size decision made on the 100-query validation split. PQ-only ordering is unchanged once the requested output depth is below every tested $K$, whereas candidate recall and exact reranking can improve as more candidates become available. Recall reaches 0.96 at $K=100$ and remains 0.96 at $K=200$. The return-raw-vector and projected-rerank rows also plateau on Hit@10 at 0.94. We therefore select $K=100$ rather than doubling encrypted work for no observed validation recall gain.

\begin{table*}[t]
\caption{Validation-only shortlist sweep on the one-million-vector self-retrieval diagnostic. Values are means over 100 queries. This table was used to freeze $K=100$ before the held-out test.}\label{tab:k-validation}
\centering
\small
\resizebox{\textwidth}{!}{%
\begin{tabular}{@{}r cccc cccc@{}}
\toprule
& \multicolumn{4}{c}{Return raw float32 vectors} & \multicolumn{4}{c}{Exact projected shortlist rerank} \\
\cmidrule(lr){2-5}\cmidrule(lr){6-9}
$K$ & Hit@1 & Hit@10 & MRR@10 & Candidate recall & Hit@1 & Hit@10 & MRR@10 & Candidate recall \\
\midrule
20  & .83 & .91 & .86333 & .92 & .83 & .92 & .86583 & .92 \\
40  & .83 & .92 & .86311 & .93 & .83 & .92 & .86144 & .93 \\
64  & .83 & .92 & .86533 & .94 & .83 & .92 & .85950 & .94 \\
80  & .83 & .93 & .86676 & .95 & .82 & .93 & .85593 & .95 \\
100 & .83 & .94 & .86926 & .96 & .82 & .94 & .85843 & .96 \\
200 & .83 & .94 & .86926 & .96 & .82 & .94 & .85819 & .96 \\
\bottomrule
\end{tabular}
}
\end{table*}

The exact projected exhaustive validation reference has Hit@1 0.83 (95\% CI 0.76--0.90), Hit@10 0.95 (0.90--0.99), MRR@10 0.86468 (0.80309--0.92301), and nDCG@10 0.88502 (0.83042--0.93577). At the selected $K$, the projected shortlist has nDCG@10 0.87818 (0.82230--0.93026), a paired change of $-0.00685$ ($[-0.01807,0.00081]$). This uncertainty is large because the validation set is small; the selection rule is the candidate-recall plateau, not a claim of equivalence.

With $K$ fixed, the next comparison isolates the three CKKS implementations.\label{subsec:kernel-results}

Table~\ref{tab:microbenchmark} compares three implementations using the same query, 100 candidate vectors, CKKS parameters, rotation-key set, and host. The total server measurement includes query deserialization, plaintext encoding, HE evaluation, ciphertext serialization, and response framing. ``HE core'' isolates the reduction evaluation; the one-response row additionally spends a median 35.02 ms on score masks, group shifts, and accumulation.

\begin{table*}[t]
\caption{Identical-parameter CKKS microbenchmark at $d'=672$, $K=100$ (20 repetitions after two warmups). Times are milliseconds; response is application bytes.}\label{tab:microbenchmark}
\centering
\small
\resizebox{\textwidth}{!}{%
\begin{tabular}{@{}l rrrrrrr@{}}
\toprule
Method & Server p50 & Server p95 & HE-core p50 & Client decrypt p50 & Client decrypt p95 & Response bytes & Max abs. error \\
\midrule
Per-candidate, 100 ct & 1689.12 & 1749.77 & 577.01 & 1011.45 & 1088.72 & 13,121,661 & $1.12\times10^{-6}$ \\
Grouped blocks, 25 ct & 433.39 & 453.88 & 145.31 & 253.00 & 329.34 & 3,280,904 & $8.82\times10^{-7}$ \\
One-response, 1 ct & \textbf{224.46} & \textbf{247.29} & 139.45 & \textbf{22.84} & \textbf{35.56} & \textbf{131,876} & $1.50\times10^{-5}$ \\
\bottomrule
\end{tabular}
}
\end{table*}

The one-response method is $7.525\times$ faster in median total server time than per-candidate CKKS and reduces response size by $99.50\times$, saving 99 ciphertext objects and 12,989,785 bytes. Comparing only HE evaluation plus the additional score-packing stage gives a smaller $3.307\times$ median speedup; the larger total improvement comes from avoiding 99 ciphertext serializations. The 25-output grouped method is $3.897\times$ faster in total server time and almost exactly $4\times$ smaller on the wire than naive CKKS; relative to that grouped response, final packing saves a further 3,149,028 bytes. This decomposition is important: the result is a packing and object-count optimization, not an unexplained cryptographic acceleration.

A back-to-back portability control reruns the same $d'=672$, $K=100$, 20-repeat script on the same physical CPU under native Windows and Ubuntu 22.04 in WSL2. Table~\ref{tab:cross-os} keeps the comparison deliberately modest. TenSEAL 0.3.16, the coefficient moduli, scale, operation counts, and 131,876-byte packed response are identical, but the surrounding runtimes are not: Windows uses Python 3.11.15 and NumPy 2.4.3, whereas WSL uses Python 3.10.12 and NumPy 2.2.6 from binary wheels. Each runtime agrees with its own NumPy reference within the stated tolerance; we do not directly compare decrypted score vectors across runtimes. The result therefore demonstrates Linux execution and correctness, not an independent hardware replication or a controlled estimate of an operating-system effect.

\begin{table*}[t]
\caption{Same-CPU Windows/WSL portability control. Times are provider milliseconds; speedup is naive p50 divided by one-response p50. The primary Windows result remains Table~\ref{tab:microbenchmark}.}\label{tab:cross-os}
\centering
\scriptsize
\resizebox{\textwidth}{!}{%
\begin{tabular}{@{}lrrrrrr@{}}
\toprule
Runtime & Packed p50 & Packed p95 & Naive p50 & Naive p95 & Speedup & Packed max error \\
\midrule
Windows 11, Python 3.11 & 195.17 & 203.93 & 1712.72 & 1793.80 & $8.775\times$ & $1.50\times10^{-5}$ \\
Ubuntu 22.04 WSL2, Python 3.10 & 169.96 & 172.89 & 607.57 & 635.38 & $3.575\times$ & $1.47\times10^{-5}$ \\
\bottomrule
\end{tabular}
}
\end{table*}

Client-side savings are also material. Decrypting one response has 22.84 ms median latency, versus 1011.45 ms for 100 objects. The final sparse-mask stage increases maximum absolute error to $1.50\times10^{-5}$, still within the looser tested tolerance of relative $3\times10^{-4}$ and absolute $3\times10^{-5}$. The maximum difference between naive and one-response decrypted scores is $1.46\times10^{-5}$. These are measured numerical facts, not a universal error bound.

\subsection{Million-vector requests and scaling}\label{subsec:unified-results}

The kernel result is encouraging, but an isolated circuit is not yet a retrieval request. The unified run freezes $d'=672$, $M=96$, eight-bit PQ, $K=100$, and the one-response kernel, then executes all 400 held-out self-query requests. Table~\ref{tab:heldout-utility} separates losses by stage. PQ ordering alone is materially worse than exact projected exhaustive ranking; exact reranking recovers most of that gap, and CKKS then closely tracks the plaintext shortlist reference.

\begin{table*}[t]
\caption{Held-out self-retrieval utility over 400 queries and one million documents. Parentheses are 95\% query-bootstrap intervals from 10,000 resamples. These are diagnostic, single-relevant-item metrics, not graded semantic relevance.}\label{tab:heldout-utility}
\centering
\scriptsize
\resizebox{\textwidth}{!}{%
\begin{tabular}{@{}l cccc@{}}
\toprule
Method & Hit@1 & Hit@10 & MRR@10 & nDCG@10 \\
\midrule
Exact projected exhaustive & .8225 (.7850,.8600) & .9275 (.9000,.9525) & .85650 (.82439,.88795) & .87355 (.84429,.90231) \\
Public PQ only & .7575 (.7150,.8000) & .8750 (.8425,.9075) & .79602 (.75866,.83267) & .81494 (.77995,.84950) \\
Plaintext projected shortlist & .8250 (.7875,.8625) & .9175 (.8900,.9450) & .85743 (.82553,.88888) & .87207 (.84227,.90133) \\
One-response CKKS shortlist & .8225 (.7850,.8600) & .9175 (.8900,.9450) & .85618 (.82441,.88772) & .87115 (.84145,.90062) \\
Return raw float32 vectors & .8325 (.7950,.8675) & .9200 (.8925,.9450) & .86106 (.82901,.89224) & .87528 (.84537,.90435) \\
\bottomrule
\end{tabular}
}
\end{table*}

Candidate recall@100 is 0.935 (95\% CI 0.910--0.9575). Relative to the plaintext projected shortlist on exactly the same candidates, CKKS changes Hit@1 by $-0.0025$ ($[-0.0075,0]$), Hit@10 by 0 ($[0,0]$), MRR@10 by $-0.00125$ ($[-0.00375,0]$), and nDCG@10 by $-0.000923$ ($[-0.002768,0]$). Relative to exhaustive projected retrieval, the plaintext shortlist nDCG change is $-0.001482$ ($[-0.005535,0.002314]$). The comparison indicates that candidate selection, not CKKS arithmetic, dominates the observed quality difference in this diagnostic.

The numerical audit covers all 40,000 candidate scores. Per query, mean absolute error averages $7.48\times10^{-6}$ (95\% CI $[7.39,7.57]\times10^{-6}$), RMSE averages $9.48\times10^{-6}$ ($[9.37,9.59]\times10^{-6}$), and maximum absolute error averages $2.30\times10^{-5}$ ($[2.27,2.32]\times10^{-5}$). The maximum over the run is $3.32\times10^{-5}$. Exact top-10 ordering is retained for 0.9625 of queries (0.9425--0.9800). These values validate the implemented score circuit on the observed distribution; they do not bound future encoders or parameter sets.

Table~\ref{tab:unified-latency} reports complete local-IPC requests. The online end-to-end measurement begins before query projection and ends after client ranking. It includes client PQ search, encryption and serialization, pipe round-trip, provider gather and CKKS work, response parsing, and decryption. It excludes cached text encoding, context setup, network transport, and the post-protocol plaintext audit.

\begin{table}[t]
\caption{Held-out local spawned-process latency for 400 requests, in ms. Quantiles are taken from whole-request or whole-stage samples, never summed.}\label{tab:unified-latency}
\centering
\small
\begin{tabular}{@{}lrrrr@{}}
\toprule
Stage & Mean & p50 & p95 & p99 \\
\midrule
Query projection & 0.251 & 0.242 & 0.337 & 0.458 \\
Client PQ scan & 25.024 & 24.905 & 28.067 & 31.753 \\
Encrypt + serialize & 12.844 & 11.038 & 18.383 & 24.814 \\
Server gather & 28.871 & 28.265 & 47.759 & 55.850 \\
Server HE evaluate & 138.733 & 138.489 & 152.270 & 167.164 \\
Server score pack & 36.121 & 36.067 & 42.524 & 46.953 \\
Server total & 240.693 & 239.025 & 273.936 & 300.501 \\
Client decrypt/decode & 8.026 & 7.498 & 10.500 & 18.949 \\
\textbf{Online end-to-end} & \textbf{287.564} & \textbf{285.530} & \textbf{323.502} & \textbf{363.976} \\
\bottomrule
\end{tabular}
\end{table}

The mean online latency has a request-bootstrap 95\% interval of 284.57--291.22 ms. Mean server total is 240.69 ms (238.37--243.28 ms). The p50/p95/p99 pipe round trip is 239.54/274.47/301.02 ms, which is close to server total because both processes are local. These values must not be reported as network latency. The request size varies slightly with CKKS serialization: mean 236,281 bytes, p95 236,348, and maximum 236,391. The response is exactly 131,876 bytes for every query. Mean total application payload is 368,157 bytes. Operating-system pipe and Python object headers are excluded by the logged payload definition.

Five independent spawned-server validation restarts provide a first robustness check. Across all 500 request records, online p50/p95/p99 is 256.24/292.87/309.44 ms, with mean 258.90 ms. The five per-restart p95 values range from 278.20 to 308.27 ms; p99 ranges from 291.39 to 388.67 ms. One restart contains a 626.13 ms request and visibly increases the tail. That variability motivated the controlled saturation run in Table~\ref{tab:concurrency} rather than a single best-case throughput number.

\begin{table*}[t]
\caption{Closed-loop packed-CKKS saturation on one 10-core/16-thread host. QPS is the median of three restarts; latency pools complete pipe requests. RSS is the maximum summed worker working set and may count shared pages more than once. All 1,860 requests succeeded.}\label{tab:concurrency}
\centering
\small
\resizebox{\textwidth}{!}{%
\begin{tabular}{@{}rrrrrrr@{}}
\toprule
Concurrent clients & QPS & p50 (ms) & p95 (ms) & p99 (ms) & Peak RSS (MiB) & Effective busy cores \\
\midrule
1  & 4.255  & 213.23 & 326.57 & 358.39 & 84.2   & 0.82 \\
2  & 8.273  & 218.55 & 347.07 & 464.38 & 168.8  & 1.61 \\
4  & 14.325 & 258.22 & 378.30 & 414.36 & 338.7  & 3.39 \\
8  & 22.124 & 347.36 & 462.91 & 545.16 & 677.3  & 6.81 \\
16 & 25.993 & 568.23 & 727.43 & 846.09 & 1356.7 & 11.87 \\
\bottomrule
\end{tabular}
}
\end{table*}

Throughput grows by $6.11\times$ from one to 16 clients, while the latency and effective-core columns show saturation rather than linear scaling. The 16-client point matches the host's logical-thread count and raises p50 to 568.23 ms; adding further workers was not tested. Peak summed worker RSS grows almost linearly and remains about 1.36 GiB at that point. The naive per-candidate control reaches only 0.426 QPS at concurrency one and 2.460 QPS at concurrency eight, versus 4.255 and 22.124 QPS for packing. Thus the packed path delivers about $10.0\times$ and $9.0\times$ more throughput at those two controlled points. None of the 72 naive requests failed.

\begin{figure*}[t]
\centering
\includegraphics[width=\textwidth]{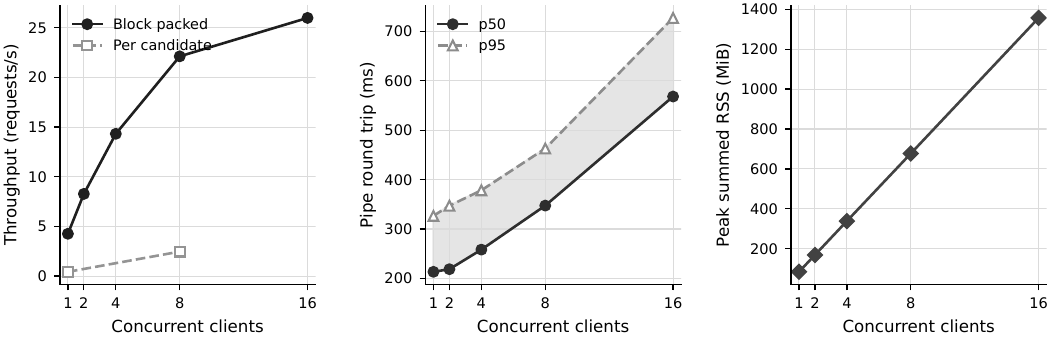}
\caption{Provider-side saturation on one physical host. Packing raises throughput substantially relative to per-candidate CKKS, but the growing p50/p95 gap and almost linear worker-memory growth make the onset of saturation visible rather than hiding it behind a single QPS value.}
\label{fig:system-scaling}
\end{figure*}

Figure~\ref{fig:system-scaling} makes the throughput--tail-latency--memory trade-off visible across the same packed conditions. Session setup is visible but amortizable. The public envelope used in the amortization calculation is 7,693,580 bytes: 80 bytes of parameters, 366,322 bytes of public key, 7,326,779 bytes of Galois material, and 399 bytes of envelope metadata. Across five fresh sessions, median time from setup start to the first decrypted result is 928.95 ms; a warm full-client request has p50/p95/p99 254.85/527.97/582.91 ms. Including that one public-context transfer, effective application bytes per request are 8,061,790 for a one-request session, 1,137,568 for ten requests, and 445,146 for 100. Fresh key epochs randomize serialized key material, so independent envelopes in the artifact range from 7,692,232 to 7,695,823 bytes; the nearby single-run sizes are not expected to be byte-identical.

The socket experiment separates local transport from HE work. Persistent TCP loopback has p50 round trip 212.36 ms and median framing/transport overhead 0.498 ms. TLS~1.3 loopback has p50 217.77 ms and overhead 0.979 ms, so the median incremental record-layer cost is about 0.48 ms after connection setup. Its one observed self-signed handshake takes 1,064.51 ms and public-context initialization another 206.67 ms; neither belongs in the warm-request number. These measurements establish a real serialized TCP/TLS implementation, not a claim about a physical network.

The same request path also reveals where the design stops scaling smoothly.\label{subsec:scaling-results} The one-response format keeps one ciphertext while $K$ increases from 20 to 200, but server work still grows with the number of four-candidate groups. Table~\ref{tab:ckks-scaling} reports the direct sweep. Response framing grows by fewer than 900 bytes because it carries a longer score-slot map, while ciphertext count remains one. At $K=200$, median server time is 421.32 ms. The corresponding per-candidate method is 3485.44 ms, so packing remains an $8.27\times$ total-server improvement at that point.

\begin{table*}[t]
\caption{One-response CKKS scaling. Each non-headline point uses seven repetitions after two warmups; the $672/100$ point uses 20.}\label{tab:ckks-scaling}
\centering
\small
\resizebox{\textwidth}{!}{%
\begin{tabular}{@{}lrrrrrr@{}}
\toprule
Sweep & $L$ & $B$ & Groups & Server p50 (ms) & Server p95 (ms) & Response bytes \\
\midrule
$d'=192$, $K=100$ & 256 & 16 & 7 & 66.66 & 74.93 & 131,855 \\
$d'=384$, $K=100$ & 512 & 8 & 13 & 118.64 & 130.71 & 131,863 \\
$d'=672$, $K=100$ & 1024 & 4 & 25 & 224.46 & 247.29 & 131,876 \\
$d'=768$, $K=100$ & 1024 & 4 & 25 & 230.47 & 240.59 & 131,876 \\
\midrule
$d'=672$, $K=20$ & 1024 & 4 & 5 & 56.38 & 60.07 & 131,475 \\
$d'=672$, $K=40$ & 1024 & 4 & 10 & 99.23 & 104.07 & 131,575 \\
$d'=672$, $K=64$ & 1024 & 4 & 16 & 151.01 & 164.88 & 131,695 \\
$d'=672$, $K=80$ & 1024 & 4 & 20 & 165.02 & 189.95 & 131,775 \\
$d'=672$, $K=200$ & 1024 & 4 & 50 & 421.32 & 435.02 & 132,351 \\
\bottomrule
\end{tabular}
}
\end{table*}

Dimension affects both the reduction depth and occupancy. At $d'=192$, 16 candidates occupy one ciphertext, whereas $d'=672$ and 768 allow only four. This discontinuous $L$ and $B$ behavior is why a retrieval-specific projection can improve HE cost even when its raw arithmetic reduction seems modest. The observed maximum errors for $d'=192,384,672,768$ at $K=100$ are respectively $1.12\times10^{-5}$, $1.01\times10^{-5}$, $1.50\times10^{-5}$, and $1.39\times10^{-5}$; no monotonic error claim is warranted from seven repetitions. Figure~\ref{fig:kernel-scaling} shows the occupancy break and the smoother shortlist-size growth together.

\begin{figure*}[t]
\centering
\includegraphics[width=\textwidth]{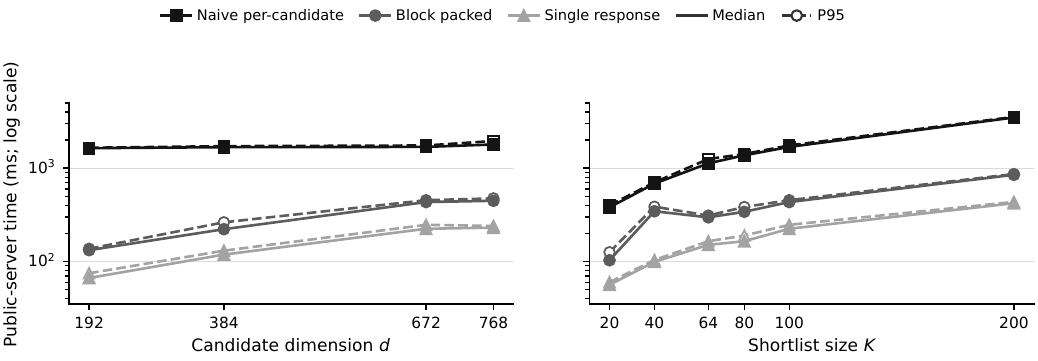}
\caption{Measured public-server work for the block-packed one-response kernel. Left: changing projected dimension at $K=100$ changes both reduction depth and ciphertext occupancy. Right: at $d'=672$, one response ciphertext remains sufficient as $K$ grows, while the number of four-candidate groups drives work. Solid lines show medians and dashed lines p95 values.}
\label{fig:kernel-scaling}
\end{figure*}

The client-side exhaustive PQ scan scales nearly linearly in the tested range. Table~\ref{tab:pq-scaling} uses 100 queries and five measured repetitions per query after five warmups. A least-squares fit of median milliseconds against document count has slope $2.4530\times10^{-5}$ ms/document, intercept 0.151 ms, and $R^2=0.99995$. This fit describes one Faiss build and host, not an asymptotic proof.

\begin{table}[t]
\caption{Client-side exhaustive PQ scan scaling for $K=100$, $M=96$. Times are ms; index bytes include codebooks and headers.}\label{tab:pq-scaling}
\centering
\small
\begin{tabular}{@{}rrrrr@{}}
\toprule
$N$ & Index bytes & p50 & p95 & p99 \\
\midrule
100,000 & 10,288,214 & 2.535 & 2.949 & 3.876 \\
250,000 & 24,688,214 & 6.315 & 7.019 & 8.910 \\
500,000 & 48,688,214 & 12.492 & 14.311 & 17.405 \\
1,000,000 & 96,688,214 & 24.642 & 27.791 & 30.711 \\
\bottomrule
\end{tabular}
\end{table}

The concurrency experiment isolates provider reranking and does not run concurrent million-entry client scans. No measured sharding or approximate coarse-quantizer index is included. Thus the data support linear local client cost, not a claim that candidate search is trivially distributed. At substantially larger $N$, persistent client storage and scan time may dominate and motivate graph search, inverted lists, PIR composition, or a different trust point.

\subsection{Vector-return baselines and measured leakage}\label{subsec:return-vector-results}

Speedup over naive HE is only half of the comparison. The more demanding alternative is to return candidate vectors directly. Table~\ref{tab:return-vector} reports warm-process post-candidate work over the same 400 frozen shortlists. These baselines gather and truly serialize the vectors, parse or dequantize them on the client, and rerank. They exclude network transfer and the preceding PQ scan, so their latency is not summed with PQ quantiles. The disclosure label is essential: float32 returns exact candidate vectors; float16 and int8 return near-exact reusable approximations.

\begin{table*}[t]
\caption{Strong vector-return baselines on 400 frozen $K=100$ test shortlists. Latency starts after candidate selection; for CKKS, projection and PQ time are removed per request before taking quantiles. Response excludes the 800-byte ID request. Vector-return methods disclose exact or near-exact vectors.}\label{tab:return-vector}
\centering
\small
\resizebox{\textwidth}{!}{%
\begin{tabular}{@{}lrrrrrr@{}}
\toprule
Returned representation & Response bytes & Post-candidate p50 & Post-candidate p95 & nDCG@10 & Top-10 set match & Disclosure \\
\midrule
Projected float32, $100\times672$ & 268,800 & 0.308 & 0.421 & .87207 & 1.0000 & exact vectors \\
Projected float16 & 134,400 & \textbf{0.261} & \textbf{0.294} & .87207 & .9925 & near-exact vectors \\
Projected row-int8 + scales & \textbf{67,600} & 0.352 & 0.443 & .87312 & .8300 & quantized vectors \\
Raw float32, $100\times768$ & 307,200 & 0.348 & 0.443 & .87528 & 1.0000 & exact vectors \\
Raw float16 & 153,600 & 0.405 & 0.455 & .87528 & .9900 & near-exact vectors \\
\midrule
One-response CKKS scores & 131,876 & 260.636 & 297.800 & .87115 & .9850 & scores only \\
\bottomrule
\end{tabular}
}
\end{table*}

Projected float16 is almost the same size as the encrypted response and has no observed IR-metric change relative to projected float32. Its mean score MAE is $4.18\times10^{-6}$ and top-10 exact-order match is 0.9775. Row-int8 reduces the response to 67,600 bytes, below CKKS, with mean score MAE $2.55\times10^{-4}$; it has top-1 match 0.99, top-10 set match 0.83, and exact-order match 0.37 relative to projected float32. Raw float16 has mean score MAE $1.14\times10^{-5}$ and exact-order match 0.95. These results prevent a misleading performance claim: vector return dominates computation and can dominate traffic. Packed CKKS is justified only when directly returning even quantized candidate vectors is outside the provider's chosen interface.

Post-candidate compute differs by less than 0.15 ms across the measured vector encodings, so the fixed protocol does not establish a meaningful performance ordering among them. The earlier integrated raw-float32 return path, which includes the million-entry PQ scan, has 24.73/28.12/30.35 ms p50/p95/p99 local end-to-end time and 307,200 response bytes; network transfer remains excluded. Either view leaves return vectors far faster locally than HE.

\begin{figure*}[t]
\centering
\includegraphics[width=\textwidth]{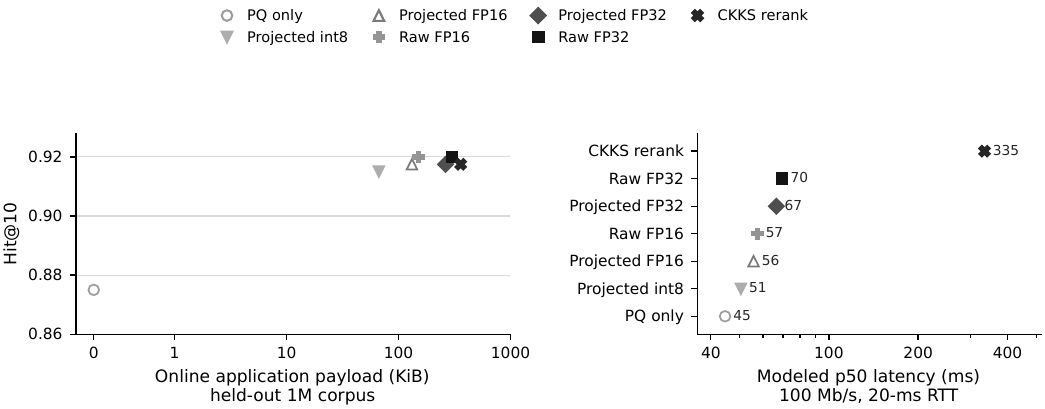}
\caption{Disclosure, payload, and ideal-link trade-off at $K=100$. Left: online application payload versus held-out Hit@10. Right: modeled p50 at 100 Mbit/s and 20 ms RTT, using each method's logged complete local p50 and payload. The model is not a network measurement; vector-return points disclose exact or near-exact candidate representations.}
\label{fig:tradeoff}
\end{figure*}

Figure~\ref{fig:tradeoff} keeps the performance and disclosure views together without converting its ideal-link calculation into a measured network result. That baseline also shows why disclosure has to be considered alongside speed.\label{subsec:pq-leakage-results} The leakage audit samples 100,000 document IDs without replacement using seed 2026. It decodes each public PQ code in projected space and also lifts the reconstruction through the public basis and passage mean to the original 768-dimensional space. Table~\ref{tab:pq-leakage} summarizes the result. Cosine is high in both spaces, especially after lifting, while relative $\ell_2$ error remains nonzero. The apparent improvement after lifting reflects the public affine reconstruction recipe; it must not be interpreted as a privacy gain.

\begin{table}[t]
\caption{Public-PQ reconstruction on 100,000 of one million documents.}\label{tab:pq-leakage}
\centering
\small
\begin{tabular}{@{}lrrrr@{}}
\toprule
Space & Mean cosine & p05 & p95 & Mean relative $\ell_2$ error \\
\midrule
Projected, 672-D & .86047 & .84061 & .88032 & .50993 \\
Lifted original, 768-D & \textbf{.96245} & .95580 & .96896 & .27117 \\
\bottomrule
\end{tabular}
\end{table}

The lifted cosine ranges from 0.93343 to 0.98521, with coordinate RMSE 0.00980. The projected cosine ranges from 0.80058 to 0.95142, with coordinate RMSE 0.01044. The audit does not run text inversion, membership inference, or attribute inference, so it cannot quantify those downstream risks. Its narrower conclusion is sufficient: a 96.7 MB public artifact is a high-fidelity lossy view of corpus geometry and must be included in the client-visible leakage.

The provider-side view is informative even before CKKS evaluation. Table~\ref{tab:candidate-leakage} fixes the deployed $K=100$ and reports both the million-vector held-out track and the frozen BEIR revision-analysis subsets. The set-only centroid recovers a direction with mean cosine between 0.3000 and 0.3897, but its more consequential use is linkage: two disjoint halves of the exposed set identify matching requests with AUC between 0.9951 and 1.0000 in this constructed audit. The order-aware estimates reproduce a substantial part of the exact neighbourhood despite never seeing a score.

\begin{table*}[t]
\caption{Semantic information retained by exposed $K=100$ candidate IDs. The provider uses only its exact vectors for those IDs. ``Set'' is the order-free centroid; the other columns use the disclosed order. Overlap is the fraction of the true exact top-10 recovered by the estimated query.}\label{tab:candidate-leakage}
\centering
\small
\resizebox{\textwidth}{!}{%
\begin{tabular}{@{}lrrrrr@{}}
\toprule
Collection & Queries & Set cosine & Set link AUC & Log-rank top-10 overlap & Ridge-rank top-10 overlap \\
\midrule
Million-vector held-out & 400   & .3897 & .999997 & .2755 & .2750 \\
ArguAna evaluable       & 1,119 & .3448 & .99716  & .5687 & .6134 \\
FiQA-2018               & 648   & .3461 & .99967  & .2907 & .3054 \\
NFCorpus                & 323   & .3149 & .99744  & .4368 & .5211 \\
SciDocs                 & 800   & .3625 & .99877  & .3916 & .4376 \\
SciFact                 & 300   & .3000 & .99512  & .4360 & .5420 \\
TREC-COVID              & 40    & .3502 & .99981  & .3475 & .1750 \\
\bottomrule
\end{tabular}
}
\end{table*}

The relevance proxy tells the same story. On the million-vector self-query track, the log-rank and ridge estimates recover the constructed source passage within top-10 for 52.0\% and 83.25\% of requests, compared with 92.75\% for the true query. On strict BEIR subsets, log-rank Hit@10 ranges from .3704 on FiQA to .9750 on TREC-COVID; ridge ranges from .4238 on SciDocs to .8250 on TREC-COVID. These numbers are not text inversion accuracy, and the source passage is easier to recognize than a free-form user intent. They nevertheless show why encrypted numerical slots cannot be promoted to semantic-query privacy once the candidate set is exposed.

The million-vector $K$ sweep also rules out a convenient monotonicity assumption. Log-rank exact top-10 overlap is .3925, .3248, .2755, and .2243 for $K=20,50,100,200$, while set-only link AUC is already .99943 at $K=20$ and approaches one thereafter. Permuting candidate IDs removes the rank signal and is still worth doing, but the set-only rows show that it does not provide unlinkability.

The return-ciphertext audit reaches a similarly narrow conclusion. Re-evaluating one fixed encrypted request 20 times produces one serialized response hash, as expected from the deterministic evaluator, whereas 20 fresh client encryptions of the same numerical query produce 20 distinct response hashes. The largest decoded score error in this diagnostic is $2.97\times10^{-5}$. Fresh input randomness therefore remains visible in ordinary correctness behaviour, but the provider adds no rerandomization, flooding, or output sanitization of its own. These observations support the explicit circuit-privacy non-claim; they are not a replacement for a proof.

The score interface is the second source of disclosure.\label{subsec:oracle-results} Rather than leave the linear-algebra limitation from Eq.~\eqref{eq:oracle-system} as a warning, we test it directly. For each of three projected rows (candidate IDs 42, 424242, and 999999), an experimental client submits the 672 standard basis vectors $e_j$ as separately encrypted $K=1$ requests. The decrypted response to $e_j$ is the $j$th coordinate of the target row. This produces 2016 actual encrypted queries through a spawned public-only server. The online client does not load exact rows; those rows are opened only after the protocol to audit reconstruction.

All three vectors are recovered almost exactly. Their cosine similarities with the exact projected rows are 0.99999999855, 0.99999999849, and 0.99999999847; relative $\ell_2$ errors are $9.96\times10^{-5}$, $9.80\times10^{-5}$, and $9.78\times10^{-5}$. Maximum coordinate errors range from $1.04\times10^{-5}$ to $1.40\times10^{-5}$. Extraction takes 30.26--31.84 s per vector and 92.87 s online in total. Per-query local-IPC latency has mean 46.07 ms and p50/p95/p99 43.89/58.03/84.69 ms because $K=1$ is much cheaper than the headline rerank.

The demonstration transfers 474,755,172 request bytes and 264,860,064 response bytes, or 739,615,236 application bytes bidirectionally; the cached 7,693,424-byte public context is excluded. It deliberately bypasses the prescribed text encoder and is therefore an adaptive malicious-client experiment, not an attack within the bounded-client claim. It does not refute confidentiality of query values from the provider. It does show that direct exact-index non-distribution does not imply document-vector or database privacy, and that rate limiting is an operational cost increase rather than a cryptographic remedy. Figure~\ref{fig:leakage} places this chosen-query result beside the two disclosure channels available before it.

\begin{figure*}[t]
\centering
\includegraphics[width=\textwidth]{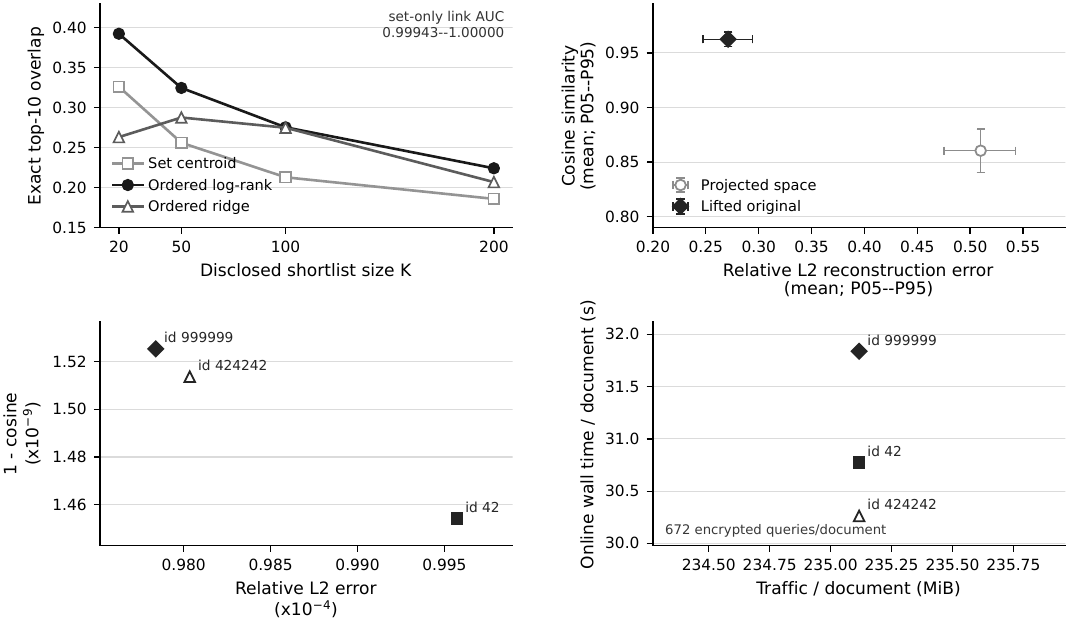}
\caption{Measured leakage is part of the system result. Candidate IDs reproduce part of the exact neighbourhood and make disjoint views almost perfectly linkable; public PQ reconstructs sampled projected and lifted vectors with substantial cosine similarity; and a malicious adaptive client recovers a fixed projected row from 672 chosen score queries. None is a claim of text inversion, but together they rule out semantic-query and database-confidentiality language for this protocol.}
\label{fig:leakage}
\end{figure*}

\subsection{Retrieval effectiveness on BEIR}\label{subsec:graded-results}

Having established what the interface costs and reveals, we finally ask whether its approximate arithmetic changes decisions on graded retrieval data. Table~\ref{tab:graded-official} reports only the frozen post-exploratory revision-analysis queries described in Section~\ref{subsec:implementation}. Projection from 768 to 672 dimensions has little effect at this operating point. PQ ordering is consistently weaker, while exact projected reranking of the same 100 candidates recovers much of the nDCG loss. Candidate Recall@100 remains visibly below exhaustive projected retrieval because reranking cannot recover a relevant document that never entered the shortlist. The CKKS column is not a plaintext surrogate: every evaluable record came from an encrypted query evaluated by a spawned public-only provider process, returned in one ciphertext, and decrypted by the client.

\begin{table*}[t]
\caption{Frozen post-exploratory BEIR revision-analysis utility at $K=100$. The first five metric columns are nDCG@10; the last two are Recall@100. ArguAna retains five absent-positive queries as zero rows in its canonical denominator. ``Actual CKKS'' summarizes the 1,119 evaluable ArguAna records and all other selected records from the full encrypted replay.}\label{tab:graded-official}
\centering
\scriptsize
\resizebox{\textwidth}{!}{%
\begin{tabular}{@{}lrrrrrrrr@{}}
\toprule
Dataset & Queries & Raw & Projected & PQ & PQ+plain & Actual CKKS & Projected R@100 & PQ R@100 \\
\midrule
SciFact & 300 & .69545 & .69678 & .66471 & .69699 & .69699 & .93600 & .91200 \\
NFCorpus & 323 & .32484 & .32703 & .29253 & .32721 & .32717 & .28654 & .26467 \\
ArguAna & 1,124 & .44434 & .44178 & .36819 & .44099 & .44050 & .95107 & .90480 \\
SciDocs & 800 & .16576 & .16614 & .13856 & .16491 & .16491 & .39206 & .32504 \\
FiQA-2018 & 648 & .38133 & .38012 & .28210 & .37774 & .37774 & .70111 & .58961 \\
TREC-COVID & 40 & .69564 & .68670 & .58035 & .69587 & .69595 & .12069 & .09288 \\
\midrule
Total & 3,235 & \multicolumn{7}{c}{3,230 evaluable actual-CKKS records selected; five canonical zero rows} \\
\bottomrule
\end{tabular}
}
\end{table*}

Table~\ref{tab:graded-ckks} isolates the CKKS replay and applies the frozen decision rule. Five collections satisfy both non-inferiority and two-sided equivalence within $\pm0.002$. ArguAna does not: its mean change is small, but the lower endpoint is $-0.002199$, just outside the margin. We therefore label that dataset inconclusive instead of converting a confidence interval that crosses zero into a claim of equivalence. The result is more limited than the earlier exploratory full-test statement, but it follows the explicitly frozen rule on query IDs disjoint from validation.

\begin{table*}[t]
\caption{Frozen post-exploratory actual-CKKS revision analysis. NI+EQ means that the paired 95\% interval is entirely inside $[-.002,.002]$; ``inconclusive'' means the criterion is not met. Latency covers evaluable local spawned-process records after candidate freezing.}\label{tab:graded-ckks}
\centering
\scriptsize
\resizebox{\textwidth}{!}{%
\begin{tabular}{@{}l rr l l rr@{}}
\toprule
Dataset & Canonical & Evaluable & $\Delta$nDCG@10 (95\% CI) & Decision & p50 (ms) & p95 (ms) \\
\midrule
SciFact & 300 & 300 & $0\ [0,0]$ & NI+EQ & 209.42 & 238.01 \\
NFCorpus & 323 & 323 & $-4.0645\times10^{-5}\ [-1.1115\times10^{-4},0]$ & NI+EQ & 211.18 & 245.18 \\
ArguAna & 1,124 & 1,119 & $-4.8915\times10^{-4}\ [-.0021990,.0011644]$ & inconclusive & 210.07 & 243.31 \\
SciDocs & 800 & 800 & $0\ [0,0]$ & NI+EQ & 210.29 & 240.75 \\
FiQA-2018 & 648 & 648 & $0\ [0,0]$ & NI+EQ & 211.34 & 242.29 \\
TREC-COVID & 40 & 40 & $+7.9425\times10^{-5}\ [0,.0002383]$ & NI+EQ & 211.17 & 229.92 \\
\bottomrule
\end{tabular}
}
\end{table*}

Across the revision-analysis subsets, p50 lies between 209.42 and 211.34 ms and p95 between 229.92 and 245.18 ms. Every response is exactly 131,876 application bytes and one SEAL ciphertext; mean request size differs by only a few bytes across datasets. The complete exploratory replay remains useful as an arithmetic and systems audit: all 3,722 evaluable requests were actually encrypted, and its largest single-query maximum score error is $2.72\times10^{-5}$. Exact top-10 order match in those full logs ranges from 0.9195 to 0.9840. Each server attestation reports a spawned process, read-only projected corpus, public context, no secret key, and no decryptor attribute. These replay latencies are lower than the million-vector end-to-end diagnostic because the graded replay deliberately starts from frozen projected queries and candidate IDs. Figure~\ref{fig:beir-utility} shows the metric decomposition and the candidate-recall ceiling without merging the six collections into one score.

\begin{figure*}[t]
\centering
\includegraphics[width=\textwidth]{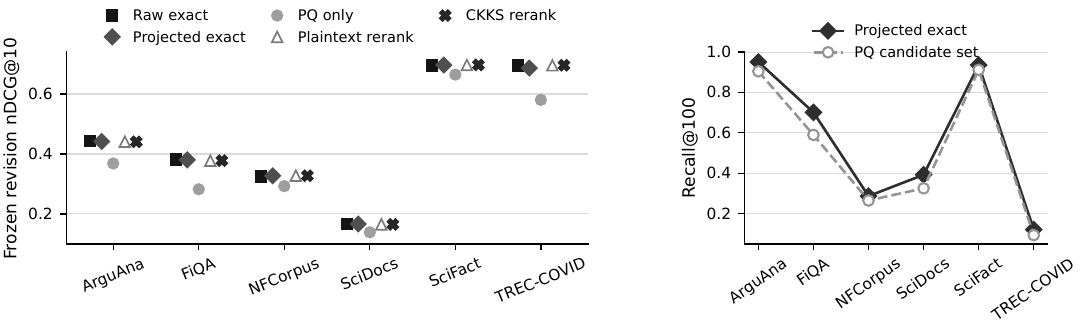}
\caption{Frozen post-exploratory BEIR revision-analysis utility of the $K=100$ public-PQ candidate set. Exact projected reranking recovers much of the PQ-ordering loss, while Recall@100 exposes the candidate ceiling. ArguAna uses the canonical denominator with five absent-positive zero rows.}
\label{fig:beir-utility}
\end{figure*}

\section{Discussion and limitations}\label{sec:discussion}

The measurements identify a specific operating point rather than a universal winner. Packing makes CKKS reranking practical under the stated boundary, but vector delivery remains faster and stronger private-retrieval systems disclose less. The useful question is therefore when this balance of latency, client storage, and visible information is acceptable.

The design is most plausible when three conditions hold. First, the client is allowed to possess a roughly 97 MB compressed view of a one-million-item index and the provider accepts the geometry leakage that follows. Second, candidate IDs may be exposed to the provider, so access-pattern privacy is not a requirement. Third, the provider regards direct delivery of candidate vectors as materially different from delivery of selected scores. Under those conditions, packed CKKS gives the provider a simple CPU circuit and the client one encrypted response.

If any condition fails, another design is preferable. If exact or quantized vector return is acceptable, Table~\ref{tab:return-vector} shows that it is decisively faster and can use less bandwidth. If the public artifact is unacceptable, this protocol does not solve the problem. If candidate IDs are sensitive, PIR/ORAM or an access-private system is required. If the provider host itself is untrusted with respect to documents, plaintext candidate vectors at the server violate the goal. The protocol is a point on a leakage/performance frontier, not an all-purpose private search layer.

The client-side artifact also changes update economics. New or deleted documents require synchronized codes and row IDs; a changed projection or PQ codebook requires a versioned rebuild. During an epoch, the provider must ensure that candidate IDs name the same exact rows the client's artifact represents. We log artifact hashes, but do not implement a transactional update protocol. Production systems would need signed manifests, rollback protection, and explicit handling of stale clients.

The performance result also needs to be read at two levels. Most of the $7.525\times$ total-server gain over naive CKKS comes from two mechanisms. Four candidate dot products share each primary homomorphic reduction, reducing 100 reductions to 25. Then one response avoids serializing 24 grouped ciphertexts or 99 naive ciphertexts. The HE-core-plus-pack gain is $3.307\times$, which is the more conservative number when transport serialization is cheap. The full server number is nevertheless relevant to an actual request because serialized ciphertext objects are the response.

The one-response stage is not free. It performs 25 additional plaintext mask multiplications and 54 extra rotations. It also raises the score scale to approximately $2^{59}$ without a second rescale and produces larger numerical error than the grouped result. This trade is favorable in the current shallow circuit and modulus chain, but it should be revalidated for any encoder dimension, scale, or backend. The abrupt occupancy changes in Table~\ref{tab:ckks-scaling} show that an apparently small dimension reduction can double $B$ and matter more than its percentage suggests.

The local end-to-end result should still be interpreted as a lower-level systems component. It contains real encryption, two processes, serialization, a full public-PQ scan, exact vector gathering, HE evaluation, and decryption. It does not contain query text encoding or a physical network. The separate TCP/TLS loopback experiment now replays the exact serialized objects over sockets and measures framing overhead, but loopback has no meaningful path RTT, loss, or congestion. Neither the 323.5-ms local end-to-end p95 nor the 217.8-ms TLS-loopback p50 is observed WAN latency. Conversely, adding a hypothetical RTT and bandwidth term to a stage p95 would not produce a valid request p95.

For planning only, the artifact provides a separately labeled ideal link model,
\begin{equation}
  T_{\mathrm{model}}=T_{\mathrm{local}}+\mathrm{RTT}+8P/R,
  \label{eq:link-model}
\end{equation}
where $P$ is logged application payload and $R$ is nominal link rate. This is not a socket measurement: it excludes TCP/TLS framing, handshake, congestion, retransmission, queueing, and final document fetch. Under 100 Mbit/s and 20 ms RTT, it maps the measured packed-CKKS p95 to 372.95 ms, versus 70.05, 59.17, and 53.98 ms for projected float32, float16, and int8 vector return. Under 1 Gbit/s and 20 ms, the corresponding modeled values are 346.45, 50.64, 49.44, and 49.05 ms. The model changes no conclusion and is not reported as observed network latency.

The disclosed channels suggest mitigations, although none turns this interface into access-private retrieval. Candidate permutation can prevent the provider from learning the PQ order but not the set. That distinction is material: the order-aware estimators improve ranking recovery, yet the set-only linkage AUC is already at least .9951 on every graded collection. Padding can conceal exact $K$ within a bucket but increases HE work. Batch timing, fixed-size framing, and key rotation can reduce some linkage channels, but repeated candidate-set overlap remains visible. Integrating PIR for candidate row access is conceptually possible, yet fetching plaintext vectors privately to an HE evaluator changes roles and can erase the simplicity of ciphertext--plaintext scoring.

The score oracle is more fundamental. Rate limiting and per-document quotas make a $672$-query linear system more costly but do not change its algebra. Constraining inputs to a certified encoder manifold might help, but proving well-formed encrypted inference would add a substantial circuit or trusted component. Returning only top-$r$ identities would reduce score leakage but requires encrypted comparison/selection or server decryption, again changing the system. Differentially private score noise could frustrate exact recovery at a utility cost. These are research directions, not hidden assumptions behind the current result.

Several limitations follow directly from this boundary. Candidate identifiers, request sizes, timing, and later payload-fetch identifiers are visible. This is not merely hypothetical metadata leakage: disjoint candidate-set views link requests almost perfectly in the measured audit, and simple estimators reproduce meaningful portions of the exact neighbourhood. The public PQ object exposes approximate corpus geometry, with lifted mean cosine 0.96245 on the measured sample, and the adaptive score interface can recover a projected row from a full-rank set of chosen queries. The returned CKKS object is not sanitized for circuit privacy. Meanwhile, the provider stores exact vectors in plaintext. The protocol consequently offers neither semantic-query privacy, document privacy, nor cryptographic database confidentiality, and it does not handle a malicious client, a compromised provider host, result substitution, integrity failures, or denial of service.

The systems evidence is broader than a single warm microbenchmark but remains local. Client storage grows linearly with the corpus, and the measured PQ scan is exhaustive; mobile deployment, sharding, transactional updates, stale-client recovery, and concurrent end-to-end client scans remain untested. Provider concurrency, warm/cold sessions, memory, persistent TCP, and TLS~1.3 are measured, but on one physical CPU and without a physical network. The WSL point exercises a Linux TenSEAL wheel on that same CPU with different Python and NumPy versions; it is a portability check, not a second-machine replication or an operating-system causal estimate. The RTX 5060 is used for embedding and reference computation, not CKKS acceleration. The loopback TLS certificate is a fixed self-signed measurement certificate, not a production authentication study, and query-encoder inference remains excluded from online latency.

Finally, CKKS is approximate. The largest observed error, $3.32\times10^{-5}$, is an empirical maximum rather than a bound for other encoders or parameter sets. Every completed lattice-estimator attack exceeds 128 bits under each stated cost model, but Arora--GB and coded-BKW timeouts prevent an exhaustive minimum, and the estimates do not prove the circular/KDM assumption for Galois material. The million-document self-retrieval setup establishes scale, not production relevance. Graded effectiveness comes from six BEIR collections; the five ArguAna qrels whose positive documents are absent from the frozen corpus remain zero rows in the canonical denominator, and that collection's frozen revision-analysis comparison is reported as inconclusive. None of these results should be generalized beyond the tested encoder, languages, corpora, parameters, and host.

These qualifications intentionally narrow the contribution. The paper asks whether block-packed encrypted reranking is practical when the candidate set and public compressed geometry are already disclosed. It does not claim that this point dominates stronger private-search systems.

\section{Conclusion}\label{sec:conclusion}

The experiments show that encrypted reranking need not inherit the obvious one-ciphertext-per-candidate cost. At $d'=672$ and $K=100$, block-SIMD grouping and final score packing reduce median server time from 1689.1 to 224.5 ms and replace 100 response ciphertexts totaling 13,121,661 bytes with one 131,876-byte object. In provider-only saturation over pre-encrypted requests, 16 isolated workers sustain 25.993 requests/s without a failed request; persistent TLS~1.3 loopback adds less than one millisecond of median framing and transport overhead after setup. The full public-only-provider replay supplies an actual CKKS record for each of 3,230 evaluable queries in the frozen post-exploratory revision subset. Against plaintext reranking of the identical candidate sets, five collections meet the frozen equivalence rule, while ArguAna remains inconclusive rather than being forced into a positive claim. The million-vector path and the SVD controls connect this kernel result to retrieval scale and explain why the selected projection is more than an arbitrary dimensionality reduction.

Taken together, the system is privacy-aware because it names and measures its leakage, not because SVD or PQ protects documents. The same candidate IDs that make the lightweight protocol possible are highly linkable and reveal part of the query neighbourhood; the public PQ codes expose corpus geometry, and an adaptive score client can recover projected rows. Under the restricted model, CKKS still keeps numerical query slots from the provider while block packing makes reranking practical. Stronger goals---hiding candidate identities, corpus geometry, evaluated-circuit information, or adaptive score access---require a different protocol rather than a broader interpretation of the present one.

\backmatter

\section*{Acknowledgements}

The author thanks the maintainers of Microsoft SEAL, TenSEAL, Faiss, E5, and BEIR for making reproducible systems evaluation possible.

\section*{Declarations}

\textbf{Funding.} No external funding was received for this study.

\textbf{Competing interests.} The author declares no competing interests.

\textbf{Ethics approval and consent to participate.} Not applicable. The study uses public text-retrieval datasets and model-derived numerical artifacts and involves no human participants or animals.

\textbf{Consent for publication.} Not applicable.

\textbf{Data availability.} SciFact, NFCorpus, ArguAna, SciDocs, FiQA-2018, and TREC-COVID are available from their cited BEIR/Hugging Face providers under the source licenses. The artifact freezes separate dataset and qrel commits, the E5 model revision, corpus/query ID hashes, and the five-query ArguAna missing-document audit. Derived embeddings, projection bases, and Faiss indexes use content-addressed cache keys and can be rebuilt; query-level retrieval and CKKS logs, manifests, and result summaries accompany the artifact.

\textbf{Code availability.} The public repository at \url{https://github.com/sergkurilenko/hybrid-privacy-aware-semantic-search} contains the source code, experiment scripts, commands for the six pinned official-test runs and six full actual-CKKS replays, cache validation, result summaries, per-query records, and hash checks for this revision.

\textbf{Author contributions.} Sergey Kurilenko conceived the study, implemented the system, conducted the experiments, analyzed the results, and wrote the manuscript.

\begingroup
\emergencystretch=2em
\Urlmuskip=0mu plus 2mu\relax
\bibliography{revised_refs}
\endgroup

\end{document}